\documentclass[11pt]{article}

\usepackage[absolute,overlay]{textpos}

\usepackage{float}
\usepackage[dvips]{graphicx}
\usepackage{psfrag}

\usepackage{amsmath,amsthm,amsfonts,amssymb,amscd}
\usepackage[table]{xcolor}
\usepackage{mathrsfs}
\usepackage{upgreek}
\usepackage{stackrel}
\usepackage{tipa}
\usepackage{accents}
\usepackage[multiple]{footmisc}

\usepackage{multirow} 
\usepackage{booktabs}

\numberwithin{equation}{section}

\newcommand{\tdb}{\textup{\textcrd}}

\newcommand{\nfi}{\varphi}
\newcommand{\cp}{\partial}

\newcommand{\bs}{\boldsymbol}
\newcommand{\lp}{\left(}
\newcommand{\rp}{\right)}
\newcommand{\td}{\textup{d}}

\newcommand{\e}{\xi}

\newcommand{\Mp}{\mathcal{M^+}}
\newcommand{\Ml}{\mathcal{M^-}}

\newcommand{\lb}{\left\lbrace}
\newcommand{\rb}{\right\rbrace}
\newcommand{\lc}{\left[}
\newcommand{\rc}{\right]}
\newcommand{\ld}{\left.}
\newcommand{\rd}{\right.}

\newcommand{\rv}{\right\vert}
\newcommand{\la}{\langle}

\newcommand{\rag}{\rangle_{g}}

\newcommand{\chr}{\Upsilon}
\newcommand{\Mpm}{\mathcal{M}^{\pm}}
\newcommand{\tdo}{\tdb}
\newcommand{\lieo}{\mathsterling}

\def\hup{h^{\sharp}}

\newcommand{\eqh}{\stackbin{\hor}{=}}
\newcommand{\eqclosh}{\stackbin{\chor}{=}}
\newcommand{\eqch}{\stackbin{\chor}{=}}

\newcommand{\eqchpm}{\stackbin{\chor^{\pm}}{=}}

\newcommand{\eqchzp}{\stackbin{\chorz^{+}}{=}}

\newcommand{\spc}{\textup{ }}

\newcommand{\ke}{\kappa_{\e}}

\newcommand{\ovl}{\overline}

\newcommand{\upm}{^{\pm}}

\newcommand{\mrg}{\hat{\Gamma}}
\newcommand{\vk}{\varkappa}

\newcommand{\nsigmapm}{\widetilde{\bs{\sigma}}{}^{\pm}}
\newcommand{\nThetap}{\widetilde{\bs{\Theta}}{}^+}
\newcommand{\nThetapm}{\widetilde{\bs{\Theta}}{}^{\pm}}

\newcommand{\nTp}{\widetilde{\bs{T}}{}^+}
\newcommand{\nTm}{\widetilde{\bs{T}}{}^-}
\newcommand{\nTpm}{\widetilde{\bs{T}}{}^{\pm}}

\newcommand{\nTstep}{\widetilde{\bs{T}}}

\newcommand{\psip}{\widetilde{\psi}{}^+}
\newcommand{\sigmap}{\widetilde{\sigma}{}^+}
\newcommand{\Thetap}{\widetilde{\Theta}{}^+}
\newcommand{\psim}{\widetilde{\psi}{}^-}
\newcommand{\sigmam}{\widetilde{\sigma}{}^-}
\newcommand{\Thetam}{\widetilde{\Theta}{}^-}
\newcommand{\Tp}{\widetilde{T}{}^+}
\newcommand{\Tm}{\widetilde{T}{}^-}

\newcommand{\sigmastep}{\widetilde{\sigma}}

\newcommand{\sigmapm}{\widetilde{\sigma}{}^{\pm}}

\newcommand{\Wpm}{\widetilde{w}{}^{\pm}}
\newcommand{\Wp}{\widetilde{w}{}^{+}}
\newcommand{\Wm}{\widetilde{w}{}^{-}}
\newcommand{\nwpm}{\widetilde{\bs{w}}{}^{\pm}}

\newcommand{\vsigma}{\widetilde{\varsigma}}
\newcommand{\vsigmapm}{\widetilde{\varsigma}^{\spc\pm}}
\newcommand{\vsigmap}{\widetilde{\varsigma}^{\spc+}}
\newcommand{\vsigmam}{\widetilde{\varsigma}^{\spc-}}
\newcommand{\nvsigmapm}{\widetilde{\bs{\varsigma}}^{\spc\pm}}

\newcommand{\Rpm}{\widetilde{R}{}^{\pm}}
\newcommand{\Rp}{\widetilde{R}{}^{+}}
\newcommand{\Rm}{\widetilde{R}{}^{-}}
\newcommand{\Xipm}{\widetilde{\Xi}{}^{\pm}}
\newcommand{\Xip}{\widetilde{\Xi}{}^{+}}
\newcommand{\Xim}{\widetilde{\Xi}{}^{-}}
\newcommand{\nXipm}{\widetilde{\bs{\Xi}}{}^{\pm}}
\newcommand{\nXip}{\widetilde{\bs{\Xi}}{}^{+}}

\newcommand{\myol}[2][8]{{}\mkern#1mu\overline{\mkern-#1mu#2}{}}

\newcommand{\hor}{\mathscr{H}}
\newcommand{\chor}{\myol{\mathscr{H}}}
\newcommand{\horz}{\mathscr{H}_0}
\newcommand{\chorz}{\myol{\mathscr{H}_0}}

\makeatletter
\newsavebox\myboxA
\newsavebox\myboxB
\newlength\mylenA

\newcommand*\xoverline[2][0.75]{%
    \sbox{\myboxA}{$\m@th#2$}%
    \setbox\myboxB\null
    \ht\myboxB=\ht\myboxA%
    \dp\myboxB=\dp\myboxA%
    \wd\myboxB=#1\wd\myboxA
    \sbox\myboxB{$\m@th\overline{\copy\myboxB}$}
    \setlength\mylenA{\the\wd\myboxA}
    \addtolength\mylenA{-\the\wd\myboxB}%
    \ifdim\wd\myboxB<\wd\myboxA%
       \rlap{\hskip 0.5\mylenA\usebox\myboxB}{\usebox\myboxA}%
    \else
        \hskip -0.5\mylenA\rlap{\usebox\myboxA}{\hskip 0.5\mylenA\usebox\myboxB}%
    \fi}
\makeatother

\usepackage{geometry}
\newtheorem{thm}{Theorem}
\newtheorem{prop}{Proposition}
\newtheorem{lem}{Lemma}

\newtheorem{defn}{Definition}
\newtheorem{rem}{Remark}

\title{General matching across Killing horizons of zero order}

\author{Miguel Manzano\footnote{miguelmanzano06@usal.es}\hspace{0.17cm} and Marc Mars\footnote{marc@usal.es}\\ \\
Instituto de F\'{\i}sica Fundamental y Matem\'aticas, IUFFyM\\
Universidad de Salamanca\\
Plaza de la Merced s/n \\
37008 Salamanca, Spain\\
}

 \newcounter{mnotecount}

 \newcommand{\mnote}[1]
 {\protect{\stepcounter{mnotecount}}$^{\mbox{\footnotesize
 $\,\bullet$\themnotecount}}$ \marginpar{
 \raggedright\tiny\em
 $\,\bullet$\themnotecount: #1} }

\allowdisplaybreaks

\begin{document}

\maketitle

\begin{abstract}
  Null shells are a useful geometric construction to study the propagation of infinitesimally thin concentrations of massless particles or impulsive waves. After recalling the necessary and sufficient conditions obtained  in \cite{manzano2021null}  that allow for the matching of two spacetimes with null embedded hypersurfaces as boundaries, we will address the problem of matching across \textit{Killing horizons of zero order} in the case when the symmetry generators  are to be identified. The results are substantially different depending on whether the boundaries are non-degenerate or degenerate, and contain or not fixed points (in particular, in the former case the shells have zero pressure but non-vanishing energy density and energy flux in general). We will present the explicit form of the so-called step function in each situation. We will then concentrate on the case of actual Killing horizons admitting a bifurcation surface, where a complete description of the shell and its energy-momentum tensor can be obtained. We will conclude particularizing to the matching of two spacetimes with spherical, plane or hyperbolic symmetry without imposing this symmetry on the shell itself.
\end{abstract}

\section{Introduction}
Thin shells (also known as surface layers) are idealized geometrical objects introduced in General Relativity to describe concentrations of matter or energy that can be considered to be located on  a hypersurface. Depending on the causal character of the hypersurface, thin shells are called null, timelike, spacelike or mixed (when the causal character is point-dependent). The standard way of generating spacetimes containing thin shells is by matching two spacetime regions (one at each side of the shell). The matching theory in the context of General Relativity is well-developed and has received contributions from many authors.
Key milestones are the seminal work by Darmois \cite{darmois1927memorial} in the  timelike and spacelike cases and its extensions to the null case  by Barrab{\'e}s-Israel \cite{barrabes1991thin} (see also \cite{poisson2002reformulation} for a useful reformulation).
The general causal character was studied in \cite{mars1993geometry}, \cite{mars2007lorentzian}. 

The Darmois matching formalism for non-null boundaries consists of joining  two spacetimes $\lp\Mpm,g^{\pm}\rp$ with differentiable boundaries $\Omega^{\pm}$ by providing an identification between the boundary points and between the full tangent spaces on $\Omega^{\pm}$. For the matching to be possible, the spacetimes are required to verify the so-called shell (or preliminary) junction conditions (see e.g. \cite{clarke1987junction}, \cite{mars1993geometry}) which force the boundaries to be isometric with respect to their induced metrics. The resulting spacetime satisfies the Israel equations \cite{israel1966singular} 
and describes a thin shell whose matter content is directly linked to the jump of the extrinsic curvatures. In the general causal character case, the shell junction conditions require that the identification of the boundaries maps the corresponding first fundamental forms and that the identification of the full tangent spaces fulfils a suitable orientation requirement \cite{mars2007lorentzian}. The Israel equations for shells of arbitrary causal character were first obtained in \cite{mars2013constraint}. Expressions for the jumps of other curvature components across matching hypersurfaces of arbitrary causal character were derived in \cite{senovilla2018equations} by means of the formalism of tensor distributional calculus.

Many explicit examples of null shells in specific situations have been discussed in the literature, often by imposing additional symmetries e.g. spherical symmetry. We refer to 
    \cite{Penrose:1972xrn}
	\cite{podolsky1999nonexpanding}, 
	\cite{barrabes2003singular}, 
	\cite{podolsky2017penrose}, 
	\cite{podolsky2019cut},
    \cite{carballo2021inner}, 
    \cite{bhattacharjee2020memory},
    \cite{nikitin2018stability},
    \cite{chapman2018holographic},
    \cite{binetruy2018closed},
    \cite{kokubu2018energy},
    \cite{fairoos2017massless} 
and references therein for examples. 
In the previous paper \cite{manzano2021null}, we determined the necessary and sufficient conditions that allow for the matching of two general spacetimes with null boundaries. In order to address the problem of matching, we made use of the so-called \textit{hypersurface data} formalism \cite{mars2013constraint}, \cite{mars2020hypersurface} with which one can abstractly analyse hypersurfaces of arbitrary signature in pseudo-riemannian manifolds. For embedded hypersurfaces, the ambient metric, the embedding and a choice of transversal vector (the so-called {\it rigging vector}) determines the metric data. In this framework, the junction conditions simply impose that the metric hypersurface data of $\Omega^{\pm}$ must coincide and, in the null case, the central objects on which the matching depends are a diffeomorphism $\Psi$ between the set of null generators of $\Omega^{\pm}$ and the so-called step function $H$, which geometrically corresponds to a shift along the null generators. Moreover, the matching requires that any spatial section on $\Omega^-$ is isometric to its corresponding image in $\Omega^+$. The energy-momentum contents of the resulting shell can be explicitly written in terms the geometry of the ambient spaces and the identification of the boundaries. The fundamental results in \cite{manzano2021null} required for the present work are summarized, and expanded in certain directions, in Section \ref{secPreliminaries}.

Generically, given any two spacetimes $\lp\Mpm,g^{\pm}\rp$ with null boundaries it will not be possible to join them. When the matching is permitted, there will exist (in general) one unique way of matching, i.e. only one suitable identification of the boundary points and the full tangent spaces that makes the metric hypersurface data of both sides agree. However, as discussed in Section 4 in \cite{manzano2021null}, there are some cases in which multiple (even infinite) matchings are feasible. This particularly occurs when the boundaries $\Omega^{\pm}$ are totally geodesic embedded null hypersurfaces. Perhaps the most prominent example of totally geodesic null hypersurfaces are the Killing horizons.
In such case, the spacetimes on both sides have additional structure, namely the Killing vectors generating the horizons, and it is natural to impose that the matching preserves this symmetry, i.e. that the resulting spacetime admits a continuous vector field, which is Killing on both sides. This situation corresponds to the case when the matching process identifies the Killing vector fields with respect to which the boundaries $\Omega^{\pm}$ are Killing horizons.
This  problem includes, for example, matchings of black hole spacetimes with non-zero temperature across geodesically complete Killing horizon boundaries, or matchings across non-degenerate or degenerate Killing horizon boundaries with or without fixed points. This is the problem we set out to analyze in this paper.

Actually, we study a considerably more general situation, namely the matching across \textit{Killing horizons to order zero}. It turns out that 
the identification of a pair of preselected null tangential vector fields, one from each side, restricts severely the set of all possible step functions. For that it is not necessary that these null fields are Killing vectors. One merely needs to assume that the boundaries are totally geodesic and the set of zeroes of the preselected null fields share the basic features that the set of zeros of Killing vectors have. This is what defines the notion of
\textit{Killing horizons to order zero}
(see Definition \ref{defKHZO}). These objects are in fact closely related to the well-known concepts of \textit{non-expanding horizons}, and their particularizations of weakly isolated horizons and isolated horizons (general references are e.g.  \cite{ashtekar2000generic}, \cite{ashtekar2002geometry}, \cite{ashtekar2000isolated},  \cite{hajivcek1973exact}, \cite{jaramillo2009isolated}, \cite{mars2012stability}). While normally non-expanding horizons are diffeomorphic to $\mathbb{S}^2\times\mathbb{R}$, we shall only require that the topology of $\Omega^{\pm}$ is $S^{\pm}_{0}\times\mathbb{R}$, $S^{\pm}_{0}\subset\Omega^{\pm}$ being a spacelike submanifold. The main difference between non-expanding horizons and Killing horizons of zero order is that we preselect a null field $\xi^{\pm}$ tangent to $\Omega^{\pm}$ (called symmetry generator) such that the set of points $\mathcal{S}^{\pm}:=\{p\in\Omega^{\pm}\spc\vert\spc\xi^{\pm}\vert_p=0\}$ is either the union of smooth connected closed submanifolds of codimension two or the empty set. Non-expanding horizons are totally geodesic as a consequence of assuming that the matter model satisfies a suitable energy condition. In the setup of
Killing horizons of zero order the property of being totally geodesic is incorporated in the definition, so we can dispose of any a priori restriction on the matter model.

In this paper, we will assume constancy of the surface gravity of $\xi^{\pm}$, which is defined by $\text{grad}(g^{\pm}(\xi^{\pm},\xi^{\pm}))=-2\ke^{\pm}\xi^{\pm}$ on $\Omega^{\pm}\setminus\mathcal{S}^{\pm}$. We keep the standard terminology of calling $\Omega^{\pm}\setminus\mathcal{S}^{\pm}$ degenerate (if $\ke^{\pm}=0$) or non-degenerate (if $\ke^{\pm}\neq0$). We will see that the presence or absence of points where the null vectors $\xi^{\pm}$ vanish, as well as the causal character of $\mathcal{S}^{\pm}$ in the former case, strongly affects the matching and by extension the types of shells that can be constructed. In particular, when $\mathcal{S}^{\pm}$ are both non-empty, we will prove that they must be identified in the matching process, which concretely forces them to have the same causal character, as well as the same number of connected components. Furthermore, the matching will require the surface gravities $\ke^{\pm}$ to be either both zero or both nonzero, and the allowed step functions will take a simple, linear form. Moreover, the resulting shell necessarily has  vanishing pressure. The matching, however, still admits the following freedom: in the degenerate case, one can select two sections (one at each side) and impose their identification whereas in the non-degenerate case it is possible to choose the initial velocity along the null generators off the submanifolds $\mathcal{S}^{\pm}$ (which are now spacelike). When the boundaries are free of fixed points, the matching admits more possibilities. All of them are studied in Section \ref{secmatchingKH}, and the result is summarized in Theorem \ref{theorem} of that Section.

Once the matching across Killing horizons of order zero has been completed, we return to the case of Killing horizons. We perform an in-depth analysis of perhaps the most physically interesting situation, which occurs when the boundaries are non-degenerate Killing horizons containing bifurcation surfaces. For this particular case, we use R\'acz-Wald coordinates \cite{racz1992extensions} to derive the explicit expression of the energy-momentum contents of the shell (Theorem \ref{theoremNDYtau} in Section \ref{secRestrictionY}). We obtain that, although the pressure is zero, some effect of compression or stretching of points is taking place because the velocity along the null generators of $\chor^+$ differs from one generator to other. As a consequence we find a non-zero energy flux which points toward null generators with higher velocities. Another remarkable result is that there appears a change of sign on the energy density of the shell when the bifurcation surfaces are crossed. This is a really puzzling behaviour which seems to point out the possibility that the physical interpretation of the energy-density of the shell is not what we are used to in other physical contexts. 
Perhaps this behaviour is somehow linked to the causality change of the Killing fields from future to past across the bifurcation surface, but this is pure speculation. We stress however, that these results are fully compatible with the shell field equations obtained by Barrabés and Israel \cite{barrabes1991thin} for the case of null hypersurfaces. We even include an explicit proof of this in Section \ref{secRestrictionY}.

The paper concludes with the general matching of two spherical, plane or hyperbolic symmetric spacetimes across non-degenerate Killing horizon boundaries containing bifurcation surfaces. We work in arbitrary $(n+1)$ spacetime dimension and do not impose any a priori restrictions on the shell (in particular we do not assume that it respects the background spherical/plane/hyperbolic symmetry). The shell depends on an arbitrary positive function $\alpha$, constant along the generators (so, effectively defined on a section). The explicit expressions of the tensors determining the matter content of the shell are explicitly derived, first without imposing any restriction on the Einstein field equations and then for the specific $\Lambda$-vacuum case. The energy density depends on the Laplacian of $\alpha$ and on the jump of the ambient Ricci tensors. Whenever non-zero, the energy density unavoidably changes sign when crossing the bifurcation surface. The energy-flux depends on the gradient of $\alpha$ and it is constant along the null generators. In the $\Lambda$-vacuum cases, the matching allows for different values of $\Lambda$ on each side but fixes the jump of the mass in terms
of the jump $[\Lambda]$. An example of particular interest occurs when $[\Lambda] <0$. Then, it is possible to construct a shell of null dust (hence with vanishing energy flux) which, from a fully physically reasonable state of  positive energy density, evolves in the deterministic manner dictated by the field equations into a state with negative energy density after crossing the bifurcation surface. Shocking as this may seem, in appears to us that such state of negative density should be considered as fully physical.


The organization of the paper is as follows.
The Section  \ref{secPreliminaries} 
is divided into three parts. First, we revisit some results and identities concerning the geometry of embedded null hypersurfaces, and we present the necessary geometric objects and assumptions. We continue with a brief summary of the basic notions on the formalism of hypersurface data. In the third part we recall the results from \cite{manzano2021null} needed for this work and then we complement and expand them in several directions. With the aim of rewriting some of the results in \cite{manzano2021null}, we firstly provide identities concerning the pullback to the abstract manifold of tensor fields on the boundaries. We also analyze the behaviour of the tensor fields defining the matter content of the shell under changes of the foliations of the boundaries. In Section \ref{secKH}, we recall some well-known properties of Killing horizons, and introduce the notion of Killing horizon of zero order. Section \ref{secmatchingKH} is devoted to the actual problem of matching two spacetimes across null boundaries which are Killing horizons of zero order. We analyze separately the cases of both boundaries being degenerate, both being non-degenerate and one being degenerate and the other one non-degenerate. As already mentioned, the particular case of matching across non-degenerate Killing horizons containing bifurcation surfaces is fully addressed in Section \ref{secRestrictionY}. We conclude the main body of the paper by studying the general matching of two arbitrary spherical, plane or hyperbolic symmetric spacetimes admitting a Killing horizon with a bifurcation surface. The specific matchings of two spacetimes of Schwarzschild and Schwarzschild-de Sitter are addressed in detail. The paper finishes with two appendices. In Appendix \ref{apppullbacks}, we give the proof of the pullback identities described in the last part of Section \ref{secPreliminaries}. Appendix \ref{appedixA} establishes a geometric expression needed in Section \ref{secRestrictionY} that 
links the ambient Ricci tensor and geometric quantities at one boundary $\Omega$ when this is a non-degenerate Killing horizon with bifurcation surface. This result is known and it is included merely in order to  make the paper self-consistent (and because our derivation is very direct and simple).

\subsection{Notation}\label{secnotation}
Given a manifold $\mathcal{M}$ and a point $p\in\mathcal{M}$, the tangent and cotangent spaces at $p $ are denoted by $T_p\mathcal{M}$, $T^*_p\mathcal{M}$ respectively. As always, $T\mathcal{M}$ refers to the corresponding tangent bundle and $\Gamma\lp T\mathcal{M}\rp$ to its sections. Given an embedding $\Phi$, we use the standard notation of $\Phi^*$ and $\Phi_*$ for its pull-back and push-forward respectively. We also let $\mathcal{F}\lp\mathcal{M}\rp:=C^{\infty}\lp\mathcal{M},\mathbb{R}\rp$ and $\mathcal{F}^*\lp\mathcal{M}\rp\subset\mathcal{F}\lp\mathcal{M}\rp$ its subset of no-where zero functions. Our signature convention for Lorentzian manifolds $\lp \mathcal{M},g\rp$ is $(-,+, ... ,+)$ and we let $\nabla$ denote the Levi-Civita covariant derivative of $g$ and $g(X,Y)$ (also $\la X,Y\rag$) be the scalar product of two vector fields $X,Y\in\Gamma\lp T\mathcal{M}\rp$. Our convention of indices on an $\lp n+1\rp$-dimensional spacetime is
\begin{equation}
\nonumber \alpha,\beta,...=0,1,...,n,\quad\qquad i,j,...=1,...,n\quad\qquad I,J,...=2,3,...,n,
\end{equation}
and parenthesis (resp. brackets) denote symmetrization (resp. antisymmetrization) of indices. We write the spacetime dimension as $n+1$ and assume throughout that $n \geq 1$.

\section{Preliminaries}\label{secPreliminaries}
As indicated in the introduction, the aim of this paper is to determine the necessary and sufficient conditions that allow for the matching of two spacetimes $\lp\Mpm,g^{\pm}\rp$ across Killing horizons of zero order such that the symmetry generators get identified. This section is devoted to introducing several background notions and results needed later. Further details can be found in \cite{manzano2021null}.

\subsection{Geometry of embedded null hypersurfaces}\label{secgeomnullhyp}
We start by recalling general properties of null hypersurfaces (general references are \cite{galloway2004null}, \cite{gourgoulhon20063+}). We begin with the standard notion of embedded null hypersurface.
\begin{defn}\label{defembhyp}
(Embedded null hypersurface) Let $\lp\mathcal{M},g\rp$ be an $\lp n+1\rp-$dimensional spacetime and $\Sigma$ a manifold of dimension $n$. An embedded null hypersurface is a subset $\Omega\subset\mathcal{M}$ satisfying that there exists an embedding $\Phi:\Sigma\rightarrow\mathcal{M}$ such that $\Phi\lp \Sigma\rp=\Omega$ and that the first fundamental form $\gamma:=\Phi^*\lp g\rp$ of $\Sigma$ is degenerate.
\end{defn}
We define null generator $k$ of $\Omega$ as any null nowhere zero vector field which is tangent to $\Omega$ everywhere. Since there exists one unique degenerate direction along $\Omega$, all null generators are proportional to each other. Besides, they are pre-geodesic, i.e. for any null generator $k$, there exists a function $\kappa_k\in\mathcal{F}\lp\Omega\rp$ named surface gravity such that $\nabla_kk=\kappa_kk$. One can always find a null generator $k$ which is in addition affine (i.e. with $\kappa_k$) \cite{galloway1999maximum}, \cite{gourgoulhon20063+}. In the following and with full generality, we let $k$ be future and satisfy this property. 

The following definitions will be useful for our purposes.
\begin{defn}
\label{def1}
(Spacelike section, tangent plane and foliation of $\Omega$) Let $\Omega$ be an embedded null hypersurface and $k$ be an affine future null generator. Assume the existence of a section $\widetilde{S}\subset\Omega$ and let $s\in\mathcal{F}\lp\Omega\rp$ be the solution of the equation $k\lp s\rp=1$ with initial condition $s\vert_{\widetilde{S}}=0$. Then, the section $S_{s_0}$ is defined as the subset
\begin{equation}
S_{s_0}:=\lb p\in\Omega\textup{ }\vert\textup{ }s\lp p\rp=s_0,\textup{ }s_0\in\mathbb{R}\rb.
\end{equation}
Given $p \in \Omega$ and the section $S_{s\lp p\rp} \subset \Omega$, the tangent plane $T_pS_{s\lp p\rp}$ is defined as
\begin{equation}
T_pS_{s\lp p\rp}:=\lb X\in T_p\Omega\textup{ }\vert\textup{ } X\lp s\rp=0\rb.
\end{equation}
The family of spacelike sections $\lb S_s\rb$ define a foliation of $\Omega$ given by the levels of $s$, i.e. the subsets of constant $s$.
\end{defn}
By construction, $s$ increases towards the future. In addition to the existence of $\widetilde{S}$, we also require that all the level sets of the function $s$, i.e. the sections $\{ S_s\}$, are diffeomorphic to each other. 
As discussed in \cite{manzano2021null}, the requirements above amount to assume that the topology of $\Omega$ is $S_{s_0} \times \mathbb{R}$, with the null generators along the direction of $\mathbb{R}$. This global restriction will be henceforth supposed. It is worth stressing, however, that $s$ always exists on sufficiently small open sets, so all the results below always apply on such sets. We construct a basis $\{L,k,v_I\}$ of $\Gamma\lp T\mathcal{M}\rp\vert_{\Omega}$ satisfying the following properties:
\begin{equation}
\label{basis}
\begin{array}{cl}
\textup{(A)} & k\textup{ is an affine future null generator.}\\
\textup{(B)} & \textup{Each }v_I\textup{ is a spacelike vector field verifying that }v_I\vert_p\in T_pS_{s\lp p\rp}\textup{ at each }p\in\Omega.\\
\textup{(C)} & \textup{The basis vectors }\lb k,v_I\rb\textup{ are such that }\lc k,v_I\rc=0\textup{ and }\lc v_I,v_J\rc=0.\\
\textup{(D)} & L\textup{ is a future null vector field everywhere transversal to }\Omega.
\end{array}
\end{equation}
Consider a point $p_0\in\Omega$, the section $S_{s\lp p_0\rp}$, a point $p \in S_{s\lp p_0\rp}$ and two vectors $Z,W\in T_pS_{s\lp p_0\rp}$. The (positive definite) induced metric $h$ of $S_{s\lp p_0\rp}$ at $p$ is $ h\lp Z,W\rp\vert_p\equiv\la Z,W\rag\vert_p$. We let $\hup$ be its associated contravariant metric. Given a basis $\lb v_I\vert_p\rb$ of $T_p S_{s\lp p_0\rp}$ and its corresponding dual $\lb \omega^I\vert_p\rb$, the components of $h$ and $\hup$ are denoted by $h_{IJ}$ and $h^{IJ}$ respectively. Capital Latin indices will be raised and lowered with these metrics.  

Given a basis $\{L,k,v_I\}$, we define $n$ scalar functions
\begin{equation}
\label{eqA30}
\nfi:=-\la L,k\rag>0,\quad \qquad \psi_I:=- \la L,v_I\rag,
\end{equation}
on $\Omega$, 
the one-form $\bs{\sigma}_L\lp Z\rp\vert_p:=\frac{1}{\nfi}\la \nabla_Zk,L\rag\vert_p$, the 2-covariant tensor $\bs{\Theta}^L\lp Z, W\rp\vert_p:=\la \nabla_ZL,W\rag\vert_p$ and the second fundamental form $\bs{\chi}^k\lp Z,W\rp\vert_p\equiv\la \nabla_{Z}k,W\rag\vert_p$ of $S_{s\lp p_0\rp}$ at $p$. 
If $L$ had been chosen to be orthogonal to $S_{s(p_0)}$ (i.e. $\psi_I=0$) then $\bs{\sigma}_L$ and $\bs{\Theta}^L$ would be the torsion one-form and second fundamental form of $S_{s(p_0)}$ along $L$. However, it is convenient for our purposes to allow $L$ to be unrelated to the sections. In the general case $\bs{\sigma}_L$ and $\bs{\Theta}^L$ are generalizations of  those quantities and still encode extrinsic information of the sections. However, we emphasize that $\bs{\Theta}^L$ is {\emph{not}} symmetric in general. For later purposes, we also recall the well-known relation between the rate of change of the induced metric along $k$ and the second fundamental form of the section (see e.g. \cite{gourgoulhon20063+})
\begin{equation}
\label{k(hIJ)}
k\lp h\lp v_K,v_I\rp\rp=2\bs{\chi}^k\lp v_I,v_K\rp.
\end{equation}

The next result (Lemma 1 in \cite{manzano2021null}) gives the tangential covariant derivatives of the basis vectors.
\begin{lem}
\label{lemCovDerivatives}
Let $\Omega$ be an embedded null hypersurface, $\lb S_s\rb$ a foliation of $\Omega$ defined by $s$ and $\lb L,k,v_I\rb$ be a basis of $\Gamma\lp T\mathcal{M}\rp\vert_{\Omega}$ satisfying conditions (\ref{basis}). Then, the tangential derivatives of the basis vectors read:
\begin{align}
\label{covderA}\nabla_{v_I}v_J&=\frac{1}{\nfi}\bs{\chi}^k\lp v_I,v_J\rp L+\frac{1}{\nfi}\lp v_I\lp\psi_J\rp+\bs{\Theta}^L\lp v_I,v_J\rp-\chr_{JI}^K\psi_K\rp k+\chr_{JI}^Kv_K,\\
\label{covderB}\nabla_k{v_I}&=\nabla_{v_I}k=-\lp \bs{\sigma}_L\lp v_I\rp+\frac{1}{\nfi}\psi^B\bs{\chi}^k\lp v_I,v_B\rp\rp k+\bs{\chi}^k\lp v_I,v^B\rp v_B,\\
\label{kappanullgen}\nabla_kk&=\kappa_kk\\
\label{covderC}\nabla_{v_I}L&=\eta_I  L-\frac{1}{\nfi}\psi^J\lp \eta_I \psi_J+\bs{\Theta}^{L}\lp v_I,v_J\rp\rp k+\lp \eta_I \psi^J+\bs{\Theta}^{L}\lp v_I,v^J\rp\rp v_J,\\
\label{covderD}\nabla_kL&=\lp \dfrac{k\lp\nfi\rp}{\nfi}-\kappa_k\rp \lp L-\dfrac{\psi^I\psi_I}{\nfi}k+\psi^I v_I\rp +\lp k\lp\psi_I\rp+\nfi\bs{\sigma}_L\lp v_I\rp\rp\lp \dfrac{\psi^I}{\nfi}k-v^I\rp,
\end{align}
where $\chr_{JI}^K$ and $\eta_I$ are defined by 
\begin{align}
\label{covderinfo1} \chr_{JI}^K&:=\lp\la v^K,\nabla_{v_I}v_J\rag+\frac{1}{\nfi}\psi^K\bs{\chi}^k\lp v_I,v_J\rp\rp,\\
\label{covderinfo2} \eta_I&:=\lp \frac{1}{\nfi}v_I\lp\nfi\rp+\bs{\sigma}_L\lp v_I\rp\rp.
\end{align}
\end{lem}
\begin{rem}
A straightforward calculation based on $[v_I,v_J]=0$ yields 
\begin{equation}
\label{christoffels,s=const}\chr_{JI}^K=\dfrac{1}{2}h^{KA}\lp v_I\lp h_{AJ}\rp +v_J\lp h_{AI}\rp-v_A\lp h_{IJ}\rp\rp+\frac{1}{\nfi}\psi^K\bs{\chi}^k\lp v_I,v_J\rp.
\end{equation}
\end{rem}

\subsection{Metric hypersurface data and hypersurface data}%
The concepts of (metric) hypersurface data \cite{mars2013constraint}  \cite{mars2020hypersurface}, which we summarize next, constitute a natural framework to study the matching, as they provide the necessary setup to study hypersurfaces of arbitrary causal character from a completely abstract viewpoint. 

Let $\Sigma$ be an $n-$dimensional manifold endowed with a $2-$symmetric covariant tensor $ \gamma $, a $1-$form $\ell$ and a scalar function $\ell^{(2)}$. The four-tuple $\lb\Sigma, \gamma ,\ell,\ell^{(2)}\rb$ defines \textbf{metric hypersurface data} provided that the symmetric $2-$covariant tensor $\bs{\mathcal{A}}\vert_p$ on $T_p\Sigma\times\mathbb{R}$ defined as
\begin{equation}
\label{ambientmetric}
\begin{array}{c}
\ld\bs{\mathcal{A}}\rv_p\lp\lp W,a\rp,\lp Z,b\rp\rp:=\ld \gamma \rv_p\lp W,Z\rp+a\ld\ell\rv_p\lp Z\rp+b\ld\ell\rv_p\lp W\rp+ab \ell^{(2)}\vert_p,\\
W,Z\in T_p\Sigma, \qquad a,b\in\mathbb{R},
\end{array}
\end{equation}
has Lorentzian signature at every $p\in\Sigma$. The five-tuple $\lb \Sigma, \gamma ,\ell,\ell^{(2)},Y\rb$ defines \textbf{hypersurface data} if additionally to the metric hypersurface data $\lb \Sigma, \gamma ,\ell,\ell^{(2)}\rb$ one also has a symmetric 2-covariant tensor $Y$ on $\Sigma$.

Since $\bs{\mathcal{A}}\vert_p$ is non-degenerate, we can consider its inverse contravariant tensor $\mathcal{A}\vert_p\in T^*_p\Sigma\times\mathbb{R}$. Splitting its components as
\begin{equation}
\label{ambientinversemetric}
\begin{array}{c}
\ld\mathcal{A}\rv_p\lp\lp \bs{\alpha},a\rp,\lp \bs{\beta},b\rp\rp:=\ld P\rv_p\lp \bs{\alpha},\bs{\beta}\rp+a\ld n\rv_p\lp \bs{\beta}\rp+b\ld n\rv_p\lp \bs{\alpha}\rp+ab n^{(2)}\vert_p,\\
\bs{\alpha},\bs{\beta}\in T^*_p\Sigma, \qquad a,b\in\mathbb{R}
\end{array}
\end{equation}
defines a symmetric two-contravariant tensor (field) $P$, a vector (field) $n$ and a scalar (field) $n^{\lp 2\rp}$ on $\Sigma$. By definition of $\mathcal{A}$, it follows (see equations (3)-(6) in \cite{mars2013constraint}):
\begin{align}
\label{gammanelln} \gamma \lp n,\cdot\rp+n^{(2)}\ell&=0, & \ell\lp n\rp&=1-n^{(2)}\ell^{(2)}, \\
P\lp\cdot,\ell\rp+\ell^{(2)}n&=0, & P\lp\cdot, \gamma \lp \cdot,X\rp\rp&=X-\ell\lp X\rp n,\qquad  \gamma \lp\cdot,P\lp \cdot,\bs{\alpha}\rp\rp=\bs{\alpha}-\bs{\alpha}\lp n\rp \ell, 
\end{align}

The abstract notion of (metric) hypersurface data connects to the geometry of hypersurfaces via the concept of \textit{embedded (metric) hypersurface data} \cite{mars2020hypersurface}. The data $\lb\Sigma, \gamma ,\ell,\ell^{(2)}\rb$ is \textbf{embedded} in a spacetime $\lp \mathcal{M},g\rp$ of dimension $n+1$ if there exists an embedding $\Phi:\Sigma\hookrightarrow\mathcal{M}$ and a rigging vector field $\zeta$ along $\Phi\lp\Sigma\rp$ (i.e. a vector field which is everywhere transversal to $\Phi\lp\Sigma\rp$) satisfying 
\begin{equation}
\label{emhd}
\Phi^*\lp g\rp= \gamma , \qquad\Phi^*\lp g\lp\zeta,\cdot\rp\rp=\ell, \qquad\Phi^*\lp g\lp\zeta,\zeta\rp\rp=\ell^{(2)}.
\end{equation}
The hypersurface data $\lb\Sigma, \gamma ,\ell,\ell^{(2)},Y\rb$ is embedded if, in addition to \eqref{emhd}, it holds
\begin{equation}
\label{eqcf}
\dfrac{1}{2}\Phi^*\lp \mathscr{L}_{\zeta}g\rp=Y.
\end{equation}
Hypersurface data has a built-in gauge freedom that corresponds to the non-uniqueness of the rigging vector field. Let $\{\Sigma, \gamma ,\ell,\ell^{\lp2\rp},Y\}$ be hypersurface data, $z\in\mathcal{F}^*\lp\Sigma\rp$ and $W\in\Gamma\lp T\Sigma\rp$. The gauge transformed hypersurface data $\lb \Sigma,\mathcal{G}_{\lp z,W\rp}\lp \gamma \rp,\mathcal{G}_{\lp z,W\rp}\lp\ell\rp,\mathcal{G}_{\lp z,W\rp}\big( \ell^{(2)} \big),\mathcal{G}_{\lp z,W\rp}\lp Y\rp\rb$ is (Definition 3.1 in \cite{mars2020hypersurface}) 
\begin{align}
\label{gaugegamma&ell2} \mathcal{G}_{\lp z,W\rp}\lp \gamma \rp&:= \gamma , & 
\mathcal{G}_{\lp z,W\rp}\big(\ell^{(2)}\big)&:=z^2\lp \ell^{(2)}+2\ell\lp W\rp+ \gamma \lp W,W\rp\rp,\\ 
\label{gaugeell&Y}\mathcal{G}_{\lp z,W\rp}\lp\ell\rp&:=z\lp\ell+ \gamma \lp W,\cdot\rp\rp, & 
\mathcal{G}_{\lp z,W\rp}\lp Y\rp & :=zY+\frac{1}{2}\lp\lp\ell\otimes \tdo z+ \tdo z\otimes \ell\rp+\lieo_{zW} \gamma \rp.
\end{align}

This gauge transformation induces the following transformations on $P$, $n$ and $n^{(2)}$ (Lemma 5 in \cite{mars2013constraint}):
\begin{align}
\label{gaugePnn2a}\mathcal{G}_{\lp z,W\rp}\lp P\rp&=P+n^{(2)}W\otimes W- W\otimes n-n\otimes W,\\
\label{gaugePnn2b}\mathcal{G}_{\lp z,W\rp}\lp n\rp&=z^{-1}\lp n-n^{(2)}W\rp,\qquad\quad \mathcal{G}_{\lp z,W\rp}\big( {n}^{(2)}\big)=z^{-2}n^{(2)}.
\end{align}
Throughout this paper, the boundaries to be matched will be null hypersurfaces, which means that $n^{(2)}=0$ and that the first fundamental form $ \gamma $ is degenerate with degenerate direction $n$ (cf. \eqref{gammanelln}).

\subsection{Matching and shell junction conditions}\label{secMatching}
We  conclude the preliminaries with a summary of the matching problem across null boundaries. This problem is studied in detail in \cite{manzano2021null}. Consider two $\lp n+1\rp$-dimensional spacetimes $\lp \mathcal{M}^{\pm},g^{\pm}\rp$ with respective null hypersurfaces $\Omega^{\pm}$ as boundaries. Let $\{ L^{\pm},k^{\pm},v^{\pm}_I\big\}$ be basis of $\Gamma\lp T\mathcal{M}^{\pm}\rp\vert_{\Omega^{\pm}}$ according to \eqref{basis} and $s^{\pm}\in\mathcal{F}\lp\Omega^{\pm}\rp$ be foliation defining functions constructed as in Definition \ref{def1}.  The matching of $\lp\Mpm,g^{\pm}\rp$ requires the fulfilment of the so-called shell (also called preliminary) junction conditions (see e.g. \cite{darmois1927memorial}, \cite{obrien1952jump}, \cite{le1955theories},  \cite{israel1966singular}, \cite{bonnor1981junction}, \cite{clarke1987junction}, \cite{barrabes1991thin},  \cite{mars1993geometry}). In the language of hypersurface data, they impose \cite{mars2013constraint} that the embedded metric hypersurface data of $\Omega^{\pm}$ agree. We therefore let $\{\Sigma, \gamma ,\ell,\ell^{(2)},Y^{\pm}\}$ be hypersurface data embedded in $\lp\Mpm,g^{\pm}\rp$ with embeddings $\Phi^{\pm}$ and riggings $\zeta^{\pm}$ to be determined in the process of matching. 

Performing a matching amounts to providing an identification between the boundary points (ruled by the diffeomorphism $\bs{\Phi}:=\Phi^+\circ({\Phi^-})^{-1}:\Omega^-\longrightarrow\Omega^+$) as well as a correspondence between transverse directions given by the identification of $\zeta^{\pm}$. We hence take coordinates $\lb \lambda,y^A\rb$ on $\Sigma$, $\lambda$ being a coordinate along the degenerate direction of $ \gamma $, and introduce new basis $\{\zeta^{\pm}, e^{\pm}_1\vert_{\Phi^{\pm}\lp q\rp}= \Phi_*^{\pm}\vert_q\lp\cp_{\lambda}\rp,e^{\pm}_I\vert_{\Phi^{\pm}\lp q\rp}=\Phi_*^{\pm}\vert_q\lp\cp_{y^I}\rp\}$ to be identified in the process of matching. With full generality, one can adapt $\Phi^-$ (or $\{e_i^{-}\}$) and $\zeta^-$ to the geometric quantities already introduced on $\Omega^-$, and let all the information of the matching be contained in $\Phi^+$ and $\zeta^+$. 

For simplicity, given any function $h\in\mathcal{F}\lp\Sigma\rp$, we will make the slight abuse of notation of writing $h\circ\lp{\Phi^{\pm}}\rp^{-1}\in\mathcal{F}\lp\Omega^{\pm}\rp$ also as $h$. The matching procedure implies the existence of functions $\{ H\lp\lambda,y^A\rp, h^I\lp y^A\rp\}$ satisfying the conditions that $\partial_{\lambda} H >0$ and
$\det\lp\partial_{y^J}h^I \rp \neq0$ such that the following decompositions hold  (Section 3.3. in \cite{manzano2021null}):
\begin{align}
\label{evectors1}& e^-_1=k^-, &&  e^-_I=v^-_I, && \zeta^-=L^-,&\\
\label{evectors2}& e^+_1=\beta k^+, & &e^+_I=a_Ik^++b_I^Jv^+_J,&&\zeta^+=\dfrac{1}{A}L^++Bk^++C^Kv^+_K,&
\end{align}
where
\begin{align}
\label{b_I^J}\beta&=\cp_{\lambda}H>0, & a_I&=\cp_{y^I}H, & b_I^J&=\cp_{y^I}h^J,\\
A&=\dfrac{\nfi^+}{\nfi^-}\cp_{\lambda}H>0, & B&=\dfrac{\nfi^-}{2\cp_{\lambda}H}h_+^{AB}\omega^{\lp+\rp}_A\omega^{\lp-\rp}_B, & C^I&=\dfrac{\nfi^-}{\cp_{\lambda}H}h_+^{IJ}\omega^{\lp +\rp}_J,
\end{align}
and $\omega^{\lp\pm\rp}_A:=(b^{-1})^I_A \big( \cp_{y^I}H  - \frac{1}{\nfi^-}\psi^-_I\cp_{\lambda} H  \big) \pm \frac{1}{\nfi^+} \psi^+_A$. Observe that the coefficients $b_I^J$ do not depend on the coordinate $\lambda$.

All the matching information is therefore encoded in $\{H\lp\lambda,y^A\rp,h^{I}\lp y^A\rp\}$. Besides, $\lambda$ and $H\lp\lambda,y^A\rp$ satisfy
\begin{align}
\label{expforHs}
\{s^-\circ\Phi^-=\lambda,\spc\spc s^+\circ\Phi^+=H\}\textup{ on }\Sigma.
\end{align}
The function $H\lp\lambda,y^A\rp$ is named step function, since it measures a kind of jump along the null direction when crossing the matching hypersurface.  
The last condition imposed by the matching procedure is that any pair of points $p\in\Omega^-$, $\bs{\Phi}\lp p\rp\in\Omega^+$ must verify that 
\begin{equation}
\label{isomcondpaper1} h_{IJ}^-\vert_p=b_I^Kb_J^Lh_{KL}^+\vert_{\bs{\Phi}\lp p\rp}.
\end{equation}
This constitutes an isometry condition between the submanifold $S_{s\lp p\rp}\subset\Omega^-$ and its corresponding image on $\Omega^+$. The identification of $e^{\pm}_1$ requires the existence of a diffeomorphism $\Psi$ between the set of null generators on both sides. The geometry of the shell is determined by the jump $[Y]:=Y^+-Y^-$ (cf. \eqref{eqcf}). In particular, the components of the energy-momentum tensor of the shell are $\tau^{11}=-( n^1)^2\gamma^{IJ}[Y_{IJ}]$, $\tau^{1I}=( n^1)^2\gamma^{IJ}[Y_{1J}]$, $\tau^{IJ}=-( n^1)^2\gamma^{IJ}[Y_{11}]$. The explicit expressions of all these quantities were obtained in \cite{manzano2021null}, and are included in the next proposition. 
\begin{prop}
\label{preliminarypropenergymomtensor} Let $\nabla^{\parallel}$ be the Levi-Civita connection on a section $\{ \lambda=\text{const.}\}\subset\Sigma$. Then the components of the tensors $Y^{\pm}$ are given by
\begin{align}
\label{preliminaryYmenos}
Y^{-}_{11} =&\textup{ } \nfi^-\lp\kappa^-_{k^-}-\dfrac{\cp_{\lambda}\nfi^-}{\nfi^-}\rp,\quad Y^{-}_{1J} =-\nfi^-\lp\bs{\sigma}^-_{L^-}(v^-_J)+ \dfrac{\nabla^{\parallel}_{J}\nfi^-}{2\nfi^-}+\dfrac{\cp_{\lambda}\psi^-_J}{2\nfi^-}  \rp,\quad Y^{-}_{IJ} = \bs{\Theta}^{L^-}_-( v_{( I}^-,v_{J)}^-),\\
\label{preliminaryY11}Y_{11}^+=&\textup{ } \nfi^{-}\lp \kappa_{k^+}^+\cp_{\lambda}H + \dfrac{\cp_{\lambda}\cp_{\lambda}H}{\cp_{\lambda}H}-\dfrac{\cp_{\lambda}\nfi^-}{\nfi^-}\rp,\\
\label{preliminaryY1J} Y_{1J}^+=&\textup{ }  \nfi^-\Bigg( \kappa_{k^+}^+\nabla_{J}^{\parallel}H-b_J^B\bs{\sigma}_{L^+}^+( v_B^+)+\dfrac{\cp_{\lambda}\cp_{y^J}H}{\cp_{\lambda}H}-\dfrac{X^L\bs{\chi}_-^{k^-}( v_J^-,v_L^-)}{\nfi^+\cp_{\lambda}H}-\dfrac{\nabla_J^{\parallel}\nfi^-}{2\nfi^-}-\dfrac{\cp_{\lambda}\psi_J^-}{2\nfi^-}\Bigg),\\
\nonumber Y^+_{IJ}=&\textup{ }\nfi^-\Bigg(\dfrac{\kappa_{k^+}^+\nabla_{I}^{\parallel}H \textup{ }\nabla_{J}^{\parallel}H}{\cp_{\lambda}H}-\dfrac{\nabla_{{(I}}^{\parallel}H\textup{ }b_{J)}^B\cp_{\lambda}\psi^+_{B}}{\nfi^+(\cp_{\lambda}H)^2}-\dfrac{2\nabla_{{(I}}^{\parallel}H\textup{ }b_{J)}^B\bs{\sigma}_{L^+}^+(v^+_{B})}{\cp_{\lambda}H} +\dfrac{b_{(I}^Ab_{J)}^B\bs{\Theta}^{L^+}_+(v_{A}^+,v_{B}^+)}{\nfi^+\cp_{\lambda}H}\\
\label{preliminaryYIJ}&\textup{ }+ \dfrac{X^1\bs{\chi}_-^{k^-}( v_I^-,v_J^-)}{\nfi^+\cp_{\lambda}H}+\dfrac{\nabla_{I}^{\parallel}\nabla_{J}^{\parallel}H}{\cp_{\lambda}H}+\dfrac{\nabla_{(I }^{\parallel}(b_{J)}^B\psi^+_{B})}{\nfi^+\cp_{\lambda}H}-\dfrac{\nabla_{( I }^{\parallel}\psi^-_{J)}}{\nfi^-}\Bigg),
\end{align}
and the components of the energy-momentum tensor of the shell are
\begin{align}
\nonumber \tau^{11}=&\textup{ }-\dfrac{\gamma^{IJ}}{\nfi^-}\Bigg( \dfrac{\kappa_{k^+}^+\nabla_{I}^{\parallel}H \textup{ }\nabla_{J}^{\parallel}H}{\cp_{\lambda}H}-\dfrac{\nabla_{{I}}^{\parallel}H\textup{ }b_J^B\cp_{\lambda}\psi^+_{B}}{\nfi^+(\cp_{\lambda}H)^2}-\dfrac{2\nabla_{{I}}^{\parallel}H\textup{ }b_J^B\bs{\sigma}_{L^+}^+(v^+_{B})}{\cp_{\lambda}H}+\dfrac{b_I^Ab_J^B\bs{\Theta}^{L^+}_+(v_{A}^+,v_{B}^+)}{\nfi^+\cp_{\lambda}H}\\
\label{preliminaryfinaltau1} &\textup{ }+ \dfrac{X^1\bs{\chi}_-^{k^-}( v_I^-,v_J^-)}{\nfi^+\cp_{\lambda}H}+\dfrac{\nabla_{I}^{\parallel}\nabla_{J}^{\parallel}H}{\cp_{\lambda}H}+\dfrac{\nabla_{I}^{\parallel}(b_J^B\psi^+_{B})}{\nfi^+\cp_{\lambda}H}-\dfrac{\nabla_{I}^{\parallel}\psi^-_{J}}{\nfi^-}- \dfrac{\bs{\Theta}^{L^-}_-( v_{I}^-,v_{J}^-)}{\nfi^-}\Bigg),\\
\label{preliminaryfinaltau2} \tau^{1I}=&\textup{ }\dfrac{\gamma^{IJ}}{{\nfi^-}}\Bigg( \kappa_{k^+}^+\nabla_{J}^{\parallel}H+\dfrac{\cp_{\lambda}\cp_{y^J}H}{\cp_{\lambda}H}-\dfrac{X^L\bs{\chi}_-^{k^-}(v_J^-,v_L^-)}{\nfi^+\cp_{\lambda}H}-\lp b_J^B\bs{\sigma}_{L^+}^+(v_B^+)-\bs{\sigma}^-_{L^-}( v^-_J)\rp \Bigg),\\
\label{preliminaryfinaltau3} \tau^{IJ}=&\textup{ }-\dfrac{\gamma^{IJ}}{{\nfi^-}}\lp  \kappa_{k^+}^+\cp_{\lambda}H-\kappa^-_{k^-} + \dfrac{\cp_{\lambda}\cp_{\lambda}H}{\cp_{\lambda}H} \rp,
\end{align}
where $\nabla_{I}^{\parallel}\psi^-_{J}$, $\nabla_{I}^{\parallel}(b_J^B\psi^+_{B})$ are derivatives of the covariant tensors $\psi^-_Idy^I$, $b_I^A\psi^+_Ady^I$ defined on the sections $\{\lambda=\text{const.}\}\subset\Sigma$ and 
\begin{align}
\label{X1XAomegapaper1}\varepsilon_I:=\nfi^+\cp_{y^I} H+b_I^J\psi_J^+,\qquad X^J:=\gamma^{IJ}\lp \varepsilon_I-A\psi^-_I\rp,\qquad X^1:=- \dfrac{X^J}{2\nfi^-A}\lp \varepsilon_J+A\psi^-_J\rp.
\end{align}
\end{prop}
In these expressions the primary geometric objects are tensors defined on the boundaries $\Omega^{\pm}$, and their transfer to the abstract hypersurface $\Sigma$ is performed via the explicit appearance of the quantities $b_{I}^J$ (in the case of the $\Omega^+$ boundary; for the other boundary $\Omega^-$ the transfer is immediate because of the fact that we are choosing the map $\Phi^-$ to be the identity between $\Sigma$ and $\Omega^-$). These expressions have the advantage that are fully explicit, which is a desirable feature when matchings of concrete spacetimes are to be performed. However, this form can also obscure the geometric interpretation of some quantities. It is therefore of interest to provide the results from Proposition \ref{preliminarypropenergymomtensor} with a more covariant interpretation. This different perspective will be particularly useful in Section \ref{secRestrictionY}, and requires some prior considerations regarding the pullback to $\Sigma$ of tensor fields and their derivatives on $\Omega^{\pm}$. 

Let $\{\bs{\omega}_{\pm}^1,\bs{\omega}_{\pm}^I\}$ be the dual basis of $\{k^{\pm},v_I^{\pm}\}$ and, for each constant $\lambda_0$, let $f_{\lambda_0}$ be the trivial embedding of the section $\{\lambda = \lambda_0 = \text{const.}\}$ onto $\Sigma$. Consider any pair of $p$-covariant tensor fields $\bs{T}^{\pm}$ on $\Omega^{\pm}$ with the property that $\bs{T}^{\pm}(...,k^{\pm},...)=0$, i.e. $\bs{T}^{\pm}=T^{\pm}_{A_1...A_p}\bs{\omega}_{\pm}^{A_1}\otimes ...\otimes \bs{\omega}_{\pm}^{A_p}$. By \eqref{evectors1}-\eqref{evectors2}, their pullback tensors $\nTpm:=(\Phi^{\pm})^*(\bs{T}^{\pm})$ on $\Sigma$ are
\begin{align}
\label{tensorremark1}\nTm&=\bs{T}^{-}(e_{I_1}^-,...,e_{I_p}^-)\text{ }dy^{I_1}\otimes...\otimes dy^{I_p}=T^{-}_{I_1...I_p}\text{ }dy^{I_1}\otimes...\otimes dy^{I_p},\\
\label{tensorremark2}\nTp&=\bs{T}^{+}(e_{I_1}^+,...,e_{I_p}^+)\text{ }dy^{I_1}\otimes...\otimes dy^{I_p}=b_{I_1}^{A_1}...b_{I_p}^{A_p}T^{+}_{A_1...A_p}\text{ }dy^{I_1}\otimes...\otimes dy^{I_p},
\end{align} 
while the corresponding pullback tensors $f_{\lambda}^*(\nTpm)$ on the sections $\{\lambda=\text{const.}\}\subset\Sigma$ are 
\begin{align}
\label{tensorremark3}f_{\lambda}^*(\nTm)&=T^{-}_{I_1...I_p}\text{ }dy^{I_1}\otimes...\otimes dy^{I_p},\qquad f_{\lambda}^*(\nTp)=b_{I_1}^{A_1}...b_{I_p}^{A_p}T^{+}_{A_1...A_p}\text{ }dy^{I_1}\otimes...\otimes dy^{I_p}.
\end{align}
In order to avoid cumbersome notation we shall still call $f_{\lambda}^*(\nTpm)$  as $\nTpm$. This slight abuse of notation is harmless since the context will make clear the precise meaning.

Let us introduce the tensor fields $\bs{\psi}^{\pm}:=\psi_I^{\pm} \bs{\omega}_{\pm}^I$, $\bs{\sigma}^{\pm}_{L^{\pm}}:=\bs{\sigma}^{\pm}_{L^{\pm}}(v_I^{\pm})\bs{\omega}_{\pm}^I$, $\bs{\Theta}_{\pm}^{L^{\pm}}:=\bs{\Theta}_{\pm}^{L^{\pm}}(v_I^{\pm},v_J^{\pm})\bs{\omega}_{\pm}^I\otimes \bs{\omega}^J_{\pm}$ and $\bs{\chi}^{k^{\pm}}_{\pm}:=\bs{\chi}^{k^{\pm}}_{\pm}(v^{\pm}_I,v^{\pm}_J)\bs{\omega}^I_{\pm}\otimes\bs{\omega}_{\pm}^J$ on $\Omega^{\pm}$ and, in accordance to the notation above, use a tilde to denote their pullbacks to $\Sigma$ and to the sections $\{\lambda = \text{const.}\}$. The jump of $\bs{T}^{\pm}$ on $\Sigma$ will be correspondingly defined as 
\begin{equation}
[\nTstep]:=(\Phi^+)^*(\bs{T}^+)-(\Phi^-)^*(\bs{T}^-)=\nTp-\nTm.
\end{equation}
As already indicated, the appearance of the coefficients $b_I^J$ in the expressions of Proposition \ref{preliminarypropenergymomtensor} is a consequence of making the pullback to $\Sigma$ or to a section $\{\lambda=\text{const.}\}\subset\Sigma$ of the corresponding tensors on $\Omega^+$. Consequently, the expressions will take a more elegant form when written in terms of tilde quantities. Moreover, the derivatives $\nabla_{I}^{\parallel}\psi^-_{J}$, $\nabla_{I}^{\parallel}(b_J^B\psi^+_{B})$ on the sections $\{\lambda=\text{const.}\}\subset\Sigma$ will also acquire a clear geometric meaning. It is immediate to rewrite all expressions in Proposition \ref{preliminarypropenergymomtensor} in terms on this new notation. In this paper we shall only need them particularized to  $\kappa^{\pm}_k=0$  
so, for the sake of brevity we only do the rewriting in this specific case.
\begin{prop}\label{propenergymomtensor}
Let $\nabla^{\parallel}$ be the Levi-Civita connection on a section $\{ \lambda=\text{const.}\}\subset\Sigma$. If $k^{\pm}$ are chosen affine (i.e. with $\kappa_k^{\pm}=0$), then the components of the tensors $Y^{\pm}$ read
\begin{align}
\label{Ymenos}\hspace{-0.3cm}Y^{-}_{11} =&\spc -\cp_{\lambda}\nfi^-,\quad Y^{-}_{1J} =-\nfi^-\sigmam_J-\dfrac{\nabla^{\parallel}_{J}\nfi^-}{2}-\dfrac{\cp_{\lambda}\psim_J}{2}  ,\quad Y^{-}_{IJ} = \Thetam_{(IJ)},\\
\label{Y11}\hspace{-0.3cm}Y_{11}^+=&\textup{ } \nfi^{-}\lp  \dfrac{\cp_{\lambda}\cp_{\lambda}H}{\cp_{\lambda}H}-\dfrac{\cp_{\lambda}\nfi^-}{\nfi^-}\rp,\\
\label{Y1J} \hspace{-0.3cm}Y_{1J}^+=&\textup{ }  \nfi^-\Bigg( -\sigmap_J+\dfrac{\cp_{\lambda}\cp_{y^J}H}{\cp_{\lambda}H}-\dfrac{X^L\widetilde{\chi}{}^-_{JL}}{\nfi^+\cp_{\lambda}H}-\dfrac{\nabla_J^{\parallel}\nfi^-}{2\nfi^-}-\dfrac{\cp_{\lambda}\psim_J}{2\nfi^-}\Bigg),\\
\label{YIJ} \hspace{-0.3cm}Y^+_{IJ}=&\textup{ }\nfi^-\Bigg(\dfrac{\nabla_{I}^{\parallel}\nabla_{J}^{\parallel}H}{\cp_{\lambda}H}-\dfrac{\nabla_{{\lp I\rd}}^{\parallel}H\textup{ }\cp_{\lambda}\psip_{\ld J\rp}}{\nfi^+\lp\cp_{\lambda}H\rp^2}-\dfrac{2\nabla_{{\lp I\rd}}^{\parallel}H\textup{ }\sigmap_{\ld J\rp}}{\cp_{\lambda}H} +\dfrac{\Thetap_{(IJ)}}{\nfi^+\cp_{\lambda}H}+\dfrac{X^1\widetilde{\chi}{}^-_{IJ}}{\nfi^+\cp_{\lambda}H}+\dfrac{\nabla_{\lp I \rd}^{\parallel}\psip_{\ld J\rp}}{\nfi^+\cp_{\lambda}H}-\dfrac{\nabla_{\lp I \rd}^{\parallel}\psim_{\ld J \rp}}{\nfi^-}\Bigg),
\end{align}
while the components of the energy-momentum tensor of the shell are
\begin{align}
\nonumber \tau^{11}=&\textup{ }-\dfrac{\gamma^{IJ}}{\nfi^-}\Bigg(\dfrac{\nabla_{I}^{\parallel}\nabla_{J}^{\parallel}H}{\cp_{\lambda}H} -\dfrac{\nabla_{{I}}^{\parallel}H\textup{ }\cp_{\lambda}\psip_{J}}{\nfi^+\lp\cp_{\lambda}H\rp^2}-\dfrac{2\nabla_{{I}}^{\parallel}H\textup{ }\sigmap_J}{\cp_{\lambda}H}+\dfrac{X^1\widetilde{\chi}{}^-_{IJ}}{\nfi^+\cp_{\lambda}H}\\
\label{finaltau1}&\textup{ }+\dfrac{\nabla_{I}^{\parallel}\psip_{J}}{\nfi^+\cp_{\lambda}H}-\dfrac{\nabla_{I}^{\parallel}\psim_{J}}{\nfi^-}+\dfrac{\Thetap_{(IJ)}}{\nfi^+\cp_{\lambda}H}- \dfrac{\Thetam_{(IJ)}}{\nfi^-}\Bigg),\\
\label{finaltau3} \tau^{1I}=&\textup{ }\dfrac{\gamma^{IJ}}{{\nfi^-}}\Bigg( \dfrac{\cp_{\lambda}\cp_{y^J}H}{\cp_{\lambda}H}-\dfrac{X^L\widetilde{\chi}{}^-_{JL}}{\nfi^+\cp_{\lambda}H}-\lc\sigmastep_J\rc \Bigg),\qquad\tau^{IJ}=-\gamma^{IJ}\dfrac{\cp_{\lambda}\cp_{\lambda}H}{\nfi^-\cp_{\lambda}H},
\end{align}
where again $X^1$ and $X^A$ are defined by \eqref{X1XAomegapaper1}.
\end{prop}
The choice of the foliation defining functions $s^{\pm}$ is highly non-unique, which makes it interesting to study the behaviour of the tensor fields $Y^{\pm}$ and the energy-momentum tensor of the shell $\tau$ under transformations of the functions $s^{\pm}$. As in the previous proposition, we restrict the discussion to the case when the null generators $k^{\pm}$ have been chosen affine, i.e. $\kappa^{\pm}_k=0$. However, the conclusions of this section also hold in general, namely when $\kappa^{\pm}_k\neq0$.

Foliations of manifolds by codimension one submanifolds always admit a freedom of reparametrization of the foliation defining function. The restriction of the generators $k^{\pm}$ being affine and future and $s^{\pm}$ satisfying $k^{\pm} (s^{\pm} ) = 1$ reduces the full reparametrization freedom to $s^{\pm} \longrightarrow q^{\pm}  s^{\pm} +  s_0^{\pm}$, with constants $s_0^{\pm}$ and $q^{\pm} >0$. This gives rise to the natural question of how expressions of Proposition \ref{propenergymomtensor} may be affected by this freedom. From the fact that the matching has been performed assuming that the map $\Phi^{-}$ is the identity (see \eqref{evectors1}), i.e. by identifying the boundary $\Omega^-$ and the abstract manifold $\Sigma$, the changes above are of different conceptual nature and hence have different effects in the minus and in  the plus sides. Concerning the $\lp\mathcal{M}^+,g^+\rp$ side, this freedom translates into multiplying $H$ by $q^+$ and shifting by $s_0^+$ (cf. \eqref{expforHs}) as well as changing $L^+$ as $L^+\longrightarrow q^+L^+$ so that $\nfi^+$ is preserved. Moreover, $k^+$ transforms as $k^+\longrightarrow \frac{1}{q^+}k^+$. It is easy to check that equations \eqref{preliminaryYmenos}-\eqref{preliminaryfinaltau3} all remain invariant under a reparametrization of this type. For that it suffices to notice that $A\longrightarrow q^+A$, $X^L\longrightarrow q^+X^L$ and $X^1\longrightarrow q^+X^1$.

The transformation $s^-\longrightarrow q^- s^-+s^-_0$ is more subtle precisely because of the trivial identification between $\Omega^-$ and $\Sigma$. By \eqref{expforHs}, we know that $s^-\longrightarrow q^- s^-+s^-_0$ induces in turn the change $\lambda\longrightarrow q^- \lambda+s^-_0$ on $\Sigma$. However, the degenerate direction of $\Sigma$ is defined by a vector field $n$ which is proportional to $\cp_{\lambda}$. Thus, scaling and shifting $\lambda$ amounts to taking a new vector field $n'=\frac{1}{q^-}n$, i.e. to perform a gauge transformation $\mathcal{G}_{(z,W)}$ with $z=q^-$ and $W=0$. Once more, condition $k^-(s^-)=1$ forces $k^-\longrightarrow\frac{1}{q^-}k^-$ which means that, in order to preserve $\nfi^-$, we again need $L^-\longrightarrow q^-L^-$. It is immediate to check that all this entails $d\lambda\longrightarrow q^-d\lambda$, $\cp_{\lambda}\longrightarrow \frac{1}{q^-}\cp_{\lambda}$, $\psi_I^-\longrightarrow q^-\psi_I^-$, $\bs{\Theta}^{L^{-}}_{-}( v^-_I,v^-_J)\longrightarrow q^-\bs{\Theta}^{L^{-}}_{-}( v^-_I,v^-_J)$, $\bs{\chi}_-^{k^-}(v_I^-,v_J^-)\longrightarrow \frac{1}{q^-}\bs{\chi}_-^{k^-}(v_I^-,v_J^-)$ and leaves $\nfi^{+}$, $\psi_I^+$, $\bs{\sigma}^{\pm}_{L^{\pm}}(v_J^{\pm})$ and $\bs{\Theta}^{L^{+}}_{+}( v^+_I,v^+_J)$ unchanged. By virtue of expressions \eqref{preliminaryYmenos}-\eqref{preliminaryYIJ}, it follows that $Y^{\pm}_{1J}$ remain unchanged, $Y^{\pm}_{11} \longrightarrow \frac{1}{q^-}Y^{\pm}_{11}$ and $Y^{\pm}_{IJ} \longrightarrow q^-Y^{\pm}_{IJ}$, as in this case $A\longrightarrow\frac{1}{q^-}A$, $X^L\longrightarrow X^L$ and $X^1\longrightarrow q^-X^1$. Since 
\begin{align}
Y^{\pm}&=Y^{\pm}_{11}d\lambda\otimes d\lambda+Y^{\pm}_{1J} \lp d\lambda\otimes dy^J+ dy^J\otimes d\lambda\rp+Y^{\pm}_{IJ}dy^I\otimes dy^J,
\end{align}
we conclude that the transformation of the tensor $Y^{\pm}$ is
$Y^{\pm}\longrightarrow q^-Y^{\pm}$, which is consistent with the gauge behaviour \eqref{gaugeell&Y}. Concerning $\tau$, it follows from \eqref{preliminaryfinaltau1}-\eqref{preliminaryfinaltau3} that the components transform as $\tau^{11}\longrightarrow q^-\tau^{11}$, $\tau^{1I}\longrightarrow\tau^{1I}$, $\tau^{IJ}\longrightarrow\frac{1}{q^-}\tau^{IJ}$ and hence the tensor 
\begin{align}
\tau&=\tau^{11}\cp_{\lambda}\otimes\cp_{\lambda}+\tau^{1J}\lp\cp_{\lambda}\otimes\cp_{y^J}+\cp_{y^J}\otimes\cp_{\lambda}\rp+\tau^{IJ}\cp_{y^I}\otimes\cp_{y^J}
\end{align}
transforms as $\tau\longrightarrow\frac{1}{q^-}\tau$ in agreement with Proposition 7 in \cite{mars2013constraint}. The tensor $\tau$ not being invariant is a consequence of the fact that there is no canonical volume form on null hypersurfaces. Each choice of rigging (or of gauge at the abstract level) gives raise to a different volume form. From the gauge transformation of such volume forms it follows that the product $\tau\bs{\eta}$ does remain invariant (see Lemma 3.5 in \cite{mars2020hypersurface}). It is therefore important to bear in mind that from a physical point of view $\tau$ is a density, i.e. a physical magnitude per unit volume, hence the necessity of this scaling behaviour under the transformation $s^{-} \longrightarrow q^-s^-+s^-_0$.

The standard physical interpretation of the components of the energy-momentum tensor is as follows (see e.g. \cite{poisson2004relativist}). In $8\pi G=c=1$ units, $\rho := \tau^{11}$ is the energy density, $j^A := \tau^{1A}$ an energy-flux and the pressure $p$ is given by $\tau^{AB}:=p\gamma^{AB}$.

\section{Killing horizons}\label{secKH}
The basic object of this paper is the so-called Killing horizon of zero order, which we define in this section. To motivate its definition, we recall first the definition and general properties of Killing horizons. General references are \cite{frolov2012black}, \cite{wald1984general}.
\begin{defn} (Killing horizon) Let $\xi$ be a Killing vector in a spacetime $\lp \mathcal{M},g\rp$. An embedded null hypersurface $\hor$ where $\xi$ is null, nowhere zero and to which $\e$ is tangent defines a Killing horizon of $\xi$.
\end{defn}
\begin{defn} (Bifurcation surface, bifurcate Killing horizon) Let $\xi$ be a non-trivial Killing vector in a spacetime $\lp \mathcal{M},g\rp$ and assume that there exists a connected spacelike codimension-two submanifold $\mathcal{S}$ of fixed points, i.e. where $\e\vert_{\mathcal{S}}=0$. This submanifold $\mathcal{S}$ is called bifurcation surface, and the set of points along all null geodesics orthogonal to $\mathcal{S}$ comprises a bifurcate Killing horizon with respect to $\e$.
\end{defn}
A Killing horizon $\hor$ may have one or several connected components. However, one can always select $\hor$ to be such that its topological closure $\chor$ is a  smooth connected (necessarily null) hypersurface without boundary. This shall be henceforth required. We will let $\mathcal{S}$ denote the set of fixed points of $\xi$ within $\chor$, i.e. $\mathcal{S}:=\{p\in\chor\spc\vert\spc\xi\vert_p=0\}$. 
The set of fixed points of a Killing vector $\xi$ is the union of connected totally geodesic closed submanifolds of even codimension (this is proven in \cite{kobayashi1958fixed}, \cite{kobayashi1969vol2} for the Riemannian and pseudo-Riemannian cases respectively). Therefore, the Killing vector $\xi$ cannot vanish on open subsets of $\chor$. We will make use of this fact later. 

Since $\e$ is null and normal along $\hor$, there exists a function $\ke\in\mathcal{F}\lp\hor\rp$, called surface gravity and defined by 
\begin{equation}
\label{kappadef}
\textup{grad}\lp\la\e,\e\rag\rp\eqh -2\ke\e\qquad\Longleftrightarrow\qquad\nabla_{\e}\e\eqh \ke\e.
\end{equation}
The surface gravity $\ke$ is constant along the null generators of $\hor$ (see e.g. \cite{frolov2012black}). For the purposes of this paper, let us also assume that $\ke$ is actually constant on $\hor$. This holds in many situations of physical interest, namely when  
(a) the Einstein tensor of the spacetime $\lp \mathcal{M},g\rp$ satisfies the dominant energy condition \cite{bardeen1973four}, \cite{wald1984general}, (b) the Killing vector field $\e$ is integrable, i.e. it verifies $\bs{\e}\wedge\td\bs{\e}=0$ \cite{racz1992extensions}, (c) for any bifurcate Killing horizon \cite{kay1991theorems}, \cite{racz1992extensions}, \cite{frolov2012black}, (d) for multiple Killing horizons \cite{mars2018multiple}, \cite{mars2018nearhorizon}, \cite{mars2019multiple}. We however emphasize that the constancy of the surface gravity restricts the class of horizons under consideration (see e.g. \cite{mars2012stability} for a situation where a non-constant surface gravity implies a rather different behaviour of the properties of the horizon).

Constancy of $\kappa_{\xi}$ on $\hor$ allows for a trivial extension to $\chor$ as the same constant. We use the standard terminology of calling $\hor$, $\chor$ degenerate (resp. non-degenerate) if $\kappa_{\xi}=0$ (resp. $\kappa_{\xi}\neq0$). Moreover, we assume $\ke$ to be positive in the non-degenerate case. Since $\ke$ is constant, this entails no loss of generality, as one can always take $-\xi$ as the Killing vector field whenever $\kappa_{\e}<0$ (cf. \eqref{kappadef}).

We shall henceforth require that all assumptions from section \ref{secPreliminaries} also apply on $\chor$. This allows us to take an affine future null generator $k$ (i.e. with $\kappa_k=0$) and a function $s\in\mathcal{F}\lp\chor\rp$ satisfying $k\lp s\rp=1$ everywhere. We construct a basis $\{L,k,v_I\}$ of $\Gamma\lp T\mathcal{M}\rp\vert_{\chor}$ according to \eqref{basis}. By uniqueness of the null generator, there is a function $F\in\mathcal{F}\lp\chor\rp$ defined by
\begin{equation}
\label{etak}
\xi\eqclosh Fk.
\end{equation}
It is clear that $F$ vanishes exactly at the fixed points of $\xi$. Equations \eqref{kappanullgen} and \eqref{kappadef} together with our choice $\kappa_{k}=0$ give
\begin{equation}
\label{eqforF}
\kappa_{\xi}\xi\eqclosh F \nabla _{k }\lp F k \rp\eqclosh  k \lp F \rp \xi \quad\Longrightarrow\quad\kappa_{\xi} \eqclosh  k \lp F \rp,
\end{equation}
where the implication holds because the set of zeros of $\xi$ always has empty interior. The solution is
\begin{equation}
\label{functionF}
F\eqclosh f+\kappa_{\xi}s,
\end{equation}
where the integration function $f\in\mathcal{F}\lp\chor\rp$ satisfies $k\lp f\rp=0$. Summarizing
\begin{equation}
\label{finaletak}
\e\eqclosh \lp f+\kappa_{\xi}s\rp k, \qquad \textup{where} \qquad k\lp f\rp=0.
\end{equation}
In the matching problem, this equation holds on both boundaries $\chor^{\pm}$. 
Since $k^{\pm}$ are proportional to $e_1^{\pm}$ (cf. \eqref{evectors1}-\eqref{evectors2})
the functions $f^{\pm}\circ\Phi^{\pm}$ (or simply $f^{\pm}$) do not depend on the coordinate $\lambda$ along $\Sigma$, i.e. $f^{\pm}=f^{\pm}\lp y^A\rp$.

We have already defined $\mathcal{S}$ as the set of fixed points of $\e$ within $\chor$. The following lemma focuses on its causal character.
\begin{lem}
\label{lemzerosets}
In the setup above, if (i) $\kappa_{\xi}\vert_{\chor}\neq0$ and $\mathcal{S}\neq\emptyset$, then $\mathcal{S}$ is given by the implicit equation $s=-\frac{f}{\ke}$ and defines a bifurcation surface. If (ii) $\kappa_{\xi}\vert_{\chor}=0$, $\mathcal{S}$ is either empty or is the union of smooth connected codimension-two null submanifolds of $\chor$ defined as the zeros of $f$. 
\end{lem}
\begin{proof}
We shall make use of the fact that $F$ cannot vanish on open submanifolds of $\chor$. Let $\{s,u^I\}$ be null coordinates on $\chor$ adapted to $k$, i.e. $k = \cp_s$. From \eqref{finaletak}, we have $f(u^I)$ and the Killing vector vanishes at points where $\kappa_{\xi} s = - f$. When $\kappa_{\xi} \neq 0$ this implies (i) at once. When $\kappa_{\xi}=0$, $\e\vert_{\chor}=fk$ and either (a) $f$ vanishes no-where along $\chor$ and hence $\mathcal{S}=\emptyset$ or (b) there exist several smooth connected codimension-two subsets $\{\mathcal{F}_{\lp i \rp}\}\subset\chor$ ($i=1,2,...$) where $f$ vanishes, and hence $\mathcal{S}\equiv\bigcup_i\mathcal{F}_{\lp i \rp}$. The fact that each connected component $\mathcal{F}_{\lp i \rp}$ is a null submanifold is a consequence of $f$ depending only on the spatial coordinates $\{u^I\}$ and not on $s$.
\end{proof}
We next analyse how the geometric quantities introduced in Section \ref{secPreliminaries} are restricted by the Killing equations $\la\nabla _X \xi, Y \rag+ \la\nabla_Y \xi, X \rag =0$, $X,Y\in\Gamma\lp T\mathcal{M}\rp\vert_{\chor}$. Combining \eqref{eqforF}, \eqref{finaletak} with expressions in Lemma \ref{lemCovDerivatives} gives
\begin{align}
\label{derxikv_I}\nabla_k \xi=\ke k,\quad \nabla_{v_I} \xi=\Big( v_I\lp F\rp  -F\big( \bs{\sigma}_L\lp v_I\rp+\frac{1}{\nfi}\psi^B\bs{\chi}^k\lp v_I,v_B\rp\big) \Big) k+F\bs{\chi}^k\lp v_I,v^B\rp v_B .
\end{align}
Clearly $\{X=k,Y=k\}$ and $\{X=k, Y=v_I\}$ satisfy the Killing equations identically. When $\{X=v_I, Y=v_J\}$, using that $\bs{\chi}^k$ is symmetric yields $F\bs{\chi}^k\lp v_I,v_J\rp=0$. Consequently, $\bs{\chi}^k$ must vanish away from the zeroes of $F$ and by continuity on $\chor$, i.e. $\chor$ is totally geodesic. This simplifies equations \eqref{covderA}-\eqref{covderB}, which now read as
\begin{align}
\label{eqA27a}\nabla_{v_I}v_J &=\frac{1}{\nfi}\lp v_I\lp\psi_J\rp+\bs{\Theta}^L\lp v_I,v_J\rp-\chr_{JI}^K\psi_K\rp k+\chr_{JI}^Kv_K,\\
\label{eqA27b}\nabla_{v_I}k&=- \bs{\sigma}_L\lp v_I\rp k,
\end{align}
where $\chr_{JI}^K=\la v^K,\nabla_{v_I}v_J\rag$. As both $\nabla_{v_I}v_J$, $\nabla_{v_I}k$ are tangent to $\chor$, the connection $\nabla$ defines a map $\nabla:\Gamma\lp T\hor\rp\times\Gamma\lp T\hor\rp\longrightarrow\Gamma\lp T\hor\rp$. This property is well-known to hold for any totally geodesic embedded null hypersurface (see e.g. \cite{ashtekar2000generic}).

We can now introduce the definition of Killing horizon of zero order.
\begin{defn}\label{defKHZO} (Killing horizon of zero order, KH$_0$) Let $\lp \mathcal{M},g\rp$ be a spacetime, $\Omega\subset\mathcal{M}$ be a smooth connected null hypersurface without boundary and $\xi \in \Gamma(T\Omega)$ a null vector. Define $\mathcal{S}:=\{p\in\Omega\spc\vert\spc\xi\vert_p=0\}$. Then $\horz:=\Omega\setminus\mathcal{S}$ is a Killing horizon to order zero if:
\begin{itemize}
\item[(a)] $\mathcal{S}$ is the union of smooth connected closed submanifolds of codimension two.
\item[(b)] $\horz$ is totally geodesic.
\end{itemize}
The \textbf{surface gravity} $\ke$ of a KH$_{0}$ is defined on $\horz$ by means of \eqref{kappadef}.
\end{defn}
From now on we shall call \textit{fixed point set} the submanifold $\mathcal{S}$ and \textit{symmetry generator} the vector $\xi$ defining a KH$_0$. All the notation introduced above for Killing horizons is kept for KH${}_0$. Note that, in particular $\Omega = \chorz$.

As we did for Killing horizons, we again assume that $\ke$ is constant on $\horz$ (and extend it to $\chorz$ as the same constant) and either positive or zero, and call $\chorz$, $\horz$ non-degenerate and degenerate respectively in each case. In addition, we require that all assumptions of Section \ref{secPreliminaries} also apply to $\chorz$. Given an affine future null generator $k$ of $\chorz$ and a function $s\in\mathcal{F}(\chorz)$ satisfying $k(s)=1$ on $\chorz$, the condition that $\xi$ does not vanish on open submanifolds of $\chorz$ together with the constancy of the surface gravity $\ke$ imply \eqref{finaletak} because \eqref{eqforF} still holds. For the same reasons, the results from Lemma \ref{lemzerosets} are also valid for KH$_0$. By definition of KH$_0$, $\bs{\chi}^k=0$ on $\horz$. By condition \textit{(a)} in Definition \ref{defKHZO}, it follows by continuity that $\bs{\chi}^k=0$ on $\chorz$. Thus, \eqref{eqA27a}-\eqref{eqA27b} are also true for KH$_0$s.

Obviously any Killing horizon $\hor$ of a Killing vector $\xi$ such that $\chor$ is a smooth connected hypersurface is also a Killing horizon to order zero.

\section{Matching across KH$_0$: symmetry generators identified}\label{secmatchingKH}
Let us now address the matching problem. We shall consider two spacetimes $\lp \Mpm,g^{\pm}\rp$ with boundaries $\chorz^{\pm}$ being the closures of two KH$_0$ with respect to the symmetry generators $\e^{\pm}$ and satisfying all the assumptions of Sections \ref{secPreliminaries} and \ref{secKH}. Our aim is to study the matching of $\lp\Mpm,g^{\pm}\rp$ across $\chorz^{\pm}$ in the case when $\xi^{\pm}$ are identified up to a multiplicative constant. 

Let $\{ L^{\pm},k^{\pm},v^{\pm}_I\}$ be a basis of $\Gamma\lp T\Mpm\rp\vert_{\chorz^{\pm}}$ according to \eqref{basis} and $s^{\pm}$ be foliation defining functions constructed following Definition \ref{def1}. We follow the same convention of \cite{manzano2021null} and let the rigging vector fields be such that $\zeta^-$ points outwards and $\zeta^+$ inwards. As discussed in \cite{manzano2021null}, this restricts the possible matchings by forcing that $\chorz^-$ lies on the future of $\lp\Ml,g^-\rp$ while $\chorz^+$ lies on the past of $\lp\Mp,g^+\rp$. We encode the freedom of rescaling the symmetry generators $\xi^{\pm}$ with a non-zero real constant $a$, i.e. we let the matching identify $\xi^{-}$ and $a\xi^{+}$. Specifically the map $\bs{\Phi}:\chorz^-\longrightarrow\chorz^+$ satisfies $\bs{\Phi}_*\lp\xi^-\rp\vert_{\chorz^+}= a\xi^+\vert_{\chorz^+}$. We recall that the combination of the definitions of $\{e^{\pm}_a\}$ and the choices \eqref{evectors1} forces $\cp_{\lambda}H>0$. Equation \eqref{etak} together with \eqref{evectors1}-\eqref{evectors2} yields
\begin{align}
a\xi^+&\eqchzp aF^+k^+\eqchzp\dfrac{aF^+e_1^+}{\cp_{\lambda}H},\\
a\xi^+&\eqchzp \bs{\Phi}_*\lp\xi^-\rp\eqchzp\bs{\Phi}_*\lp F^-k^-\rp\eqchzp\bs{\Phi}_*\lp F^-e_1^-\rp\eqchzp F^-e_1^+.
\end{align}
Identifying the symmetry generators $\xi^{-}$, $a\xi^{+}$ hence forces
\begin{equation}
\label{idKillings1}
F^-\eqchzp\dfrac{aF^+}{\cp_{\lambda}H}.
\end{equation}
The matching procedure therefore requires the submanifolds $\mathcal{S}^{\pm}$ to be mapped to each other via $\bs{\Phi}$. In \cite{manzano2021null}, we proved that a general matching across totally geodesic null boundaries allowed for an infinite set of possible step functions. This was a consequence of all leaves $\{s^{\pm}=\text{const.}\}$ being isometric to each other. Because of condition \eqref{idKillings1}, this is no longer true here. Away from the zeroes of $F^{\pm}$, the integration of \eqref{idKillings1} determines $H$ up to an integration function. This solution, however, may be difficult to find in general. The problem becomes simpler under the assumptions of $k$ being affine (which implies no loss of generality) and $\ke^{\pm}$ being constant. Then, inserting \eqref{functionF} into \eqref{idKillings1} and using \eqref{expforHs} gives
\begin{equation}
\label{eqfordH}
f^-+\kappa^-_{\xi}\lambda=\dfrac{a\big( f^++\kappa^+_{\xi}H\big)}{\cp_{\lambda}H}\quad\Longrightarrow\quad \cp_{\lambda}H=\dfrac{a\big( f^++\kappa^{+}_{\xi}H\big)}{f^{-}+\kappa^{-}_{\xi}\lambda}>0,
\end{equation}
where here, and in the following, all equations are meant to hold on the abstract manifold $\Sigma$ (unless otherwise stated). Since $f^{\pm}$ are $\lambda$-independent, \eqref{eqfordH} can be easily integrated to obtain the explicit form of $H$. Note that the second expression in \eqref{eqfordH} only holds away from the points in $(\Phi^+)^{-1}(\mathcal{S}^+)$. The value of $H$ on those points is determined by continuity. The positivity of $\cp_{\lambda}H$ forces 
\begin{equation}
\label{sign}\text{sign}\lp a\rp\text{sign}\big( f^++\kappa^{+}_{\xi}H\big)=\text{sign}\big( f^{-}+\kappa^{-}_{\xi}\lambda\big),
\end{equation}
or in more geometric terms, that both symmetry generators $\{ \xi^-,a\xi^+\}$ must be simultaneously either future or past. This of course is consistent with the fact that $\{ \xi^-,a\xi^+\}$ are to be identified. We now study separately the matching for the cases (a) $\kappa_{\xi}^{\pm}=0$, (b) $\kappa_{\xi}^{\pm}\neq0$ and (c) $\kappa_{\xi}^{-}=0$, $\kappa^{+}_{\xi}\neq0$ or $\kappa_{\xi}^{-}\neq0$, $\kappa^{+}_{\xi}=0$. 

\subsection{Case of $\e^{\pm}$ degenerate}\label{secdeg}
When $\kappa_{\xi}^{\pm}$ vanish, we know by Lemma \ref{lemzerosets} that $\mathcal{S}$ is either empty or the union of smooth connected codimension-two null submanifolds of $\chorz$. The fact that the map $\bs{\Phi}$ is a diffeomorphism forces both boundaries to have the same number of such submanifolds. 

Equation \eqref{eqfordH} with $\kappa^{\pm}_{\xi}=0$ requires $\frac{af^+\lp y^A\rp}{f^-\lp y^A\rp}>0$ and yields the explicit form of the step function,
\begin{equation}
\label{caseAidKillings}
 H\lp\lambda,y^A\rp=\dfrac{af^+\lp y^A\rp}{f^-\lp y^A\rp}\lambda+\mathcal{H}\lp y^A\rp,
\end{equation}
where $\mathcal{H}\lp y^A\rp$ is an integration function. Once we select the tuple $\{a,\xi^-,\xi^+\}$, the only remaining matching freedom is encoded in the function $\mathcal{H}\lp y^A\rp$ (the scalar functions $f^{\pm}$ are known beforehand as the spacetimes to be matched are assumed to be known, cf. \eqref{finaletak}). In order to understand this freedom, let us call ``velocity" the rate of change of $s^{\pm}$ along a null generator of $\chorz^{\pm}$. The velocity along the null generators of $\chorz^{\pm}$ is totally determined (outside of $\mathcal{S}$) by the identification of $\{\xi^-,a\xi^+\}$. However, there still exist a freedom to select any pair of sections, one on each side, and force their identification via $\bs{\Phi}$. This is the freedom encoded in the arbitrary function $\mathcal{H}\lp y^A\rp$. Note that the step function \eqref{caseAidKillings} is linear in $\lambda$. This means in particular that the most general shell that can be generated under these circumstances has vanishing pressure (cf. \eqref{finaltau3}).

\subsection{Case of $\e^{\pm}$ non-degenerate}\label{secnondeg}
We now study the case when both $\horz^{\pm}$ are non-degenerate. By Lemma \ref{lemzerosets}, we know that $\mathcal{S}^{\pm}$ are either empty or sections defined by $\mathcal{S}^{\pm}:=\{p\in\chorz^{\pm}\spc\vert\spc f^{\pm}+\kappa_{\xi}^{\pm}s^{\pm}\vert_p=0\}$. We define the submanifolds $\mathcal{\hor}_{\text{p}}^{\pm}$, $\mathcal{\hor}_{\text{f}}^{\pm}$ by
\begin{align}
\hor_{\text{p}}^{\pm}:=\lb p\in\chorz^{\pm}\spc\vert\spc f^{\pm}+\kappa_{\xi}^{\pm}s^{\pm}\vert_p<0\rb,\qquad \hor_{\text{f}}^{\pm}:=\lb p\in\chorz^{\pm}\spc\vert\spc f^{\pm}+\kappa_{\xi}^{\pm}s^{\pm}\vert_p>0\rb,
\end{align}
so that $\horz^{\pm}\equiv\hor_{\text{p}}^{\pm}\cup\hor_{\text{f}}^{\pm}$ (hence $\chorz^{\pm}\equiv\hor_{\text{p}}^{\pm}\cup\mathcal{S}^{\pm}\cup\hor_{\text{f}}^{\pm}$). Since we are assuming nothing on the geodesic completeness of $\chorz^{\pm}$, we do not exclude the cases when any of $\hor_{\text{p}}^{\pm}$, $\hor_{\text{f}}^{\pm}$ and $\mathcal{S}^{\pm}$ are empty. Note that, when $\hor^{\pm}_{\text{p}}$, $\hor^{\pm}_{\text{f}}$ are non-empty they are by definition KH$_0$s.  
For later purposes, we also introduce the non-zero constant $\widehat{\kappa}:=a\kappa^+_{\xi}(\kappa^-_{\xi})^{-1}$. Note that $\text{sign}\lp\hat{\kappa}\rp=\text{sign}\lp a\rp$ because we have chosen the orientations of $\xi^{\pm}$ such that $\ke^{\pm}>0$.

Let us start by considering the case $\mathcal{S}^{\pm}\neq\emptyset$.  As already discussed, the fixed point sets must be identified in the matching process. Let us show that in such case, \eqref{eqfordH} forces $\hat{\kappa}$ to be equal to one. We apply the l'Hôpital rule to \eqref{eqfordH} and get
\begin{equation}
\label{limithatkappa}
\lim_{\lambda\rightarrow-\frac{f^-}{\ke^-}}\cp_{\lambda}H=\lim_{\lambda\rightarrow-\frac{f^-}{\ke^-}}\frac{a\ke^+\cp_{\lambda}H}{\ke^-}=\hat{\kappa}\lim_{\lambda\rightarrow-\frac{f^-}{\ke^-}}\cp_{\lambda}H\qquad\Longleftrightarrow\qquad \hat{\kappa}=1.
\end{equation}
Thus, $a=\ke^-(\ke^+)^{-1}>0$ and equation \eqref{eqfordH} becomes $\cp_{\lambda}H=\Big(\frac{f^+}{\ke^+}+H\Big)\Big(\frac{f^{-}}{\ke^-}+\lambda\Big)^{-1}>0$. Its integration yields
\begin{equation}
\ln\left\vert\frac{f^+\lp y^A\rp}{\ke^+}+H\lp\lambda,y^A\rp\right\vert=\ln\lp\alpha\lp y^A\rp\left\vert\frac{f^{-}\lp y^A\rp}{\ke^-}+\lambda\right\vert\rp,
\end{equation}
where $\alpha\lp y^A\rp$ is a positive integration function. Equation \eqref{sign} allows us to conclude that
\begin{equation}
\label{eqforHddef}
H\lp\lambda,y^A\rp=\alpha\lp y^A\rp\lp \lambda+\dfrac{f^-\lp y^A\rp}{\ke^-}\rp-\dfrac{f^+\lp y^A\rp}{\ke^+},\qquad \alpha\lp y^A\rp>0.
\end{equation}
In combination with the results in Section \ref{secdeg} we conclude that whenever $\chorz^{\pm}$ are degenerate or contain non-empty fixed point sets any matching of $\lp\Mpm,g^{\pm}\rp$ across $\chorz^{\pm}$ in which the symmetry generators  $\{\xi^-\vert_{\chor^-},a\xi^+\vert_{\chor^+}\}$ are identified requires the surface gravities $\{\ke^-,a\ke^+\}$ to coincide. Moreover, the step function must be linear in the coordinate $\lambda$, which excludes matchings giving rise to shells with non-vanishing pressure.

It is also physically interesting to study the matchings when no null generator of $\chorz^{\pm}$ crosses any fixed point set,  i.e. when $\mathcal{S}^{\pm}$ are both empty. Integrating \eqref{eqfordH} now leads to 
\begin{equation}
\label{absolutecondition}\big\vert  f^+\lp y^A\rp+\ke^+H\lp \lambda,y^A\rp\big\vert= \alpha\lp y^A\rp\big\vert f^-\lp y^A\rp+\ke^-\lambda\big\vert^{\hat{\kappa}},
\end{equation}
where $\alpha\lp y^A\rp$ is a non-zero positive integration function. We analyse the cases $a>0$ (i.e. $\hat{\kappa}>0$) and $a<0$ (i.e. $\hat{\kappa}<0$) separately. For the former, condition \eqref{sign} gives $\text{sign}( f^++\ke^+H)=\text{sign}( f^-+\ke^-\lambda)$, which only allows for the matchings (see (a), (b) in Figure \ref{fig1})
\begin{itemize}
\item[(a)]$\lb\begin{array}{l}
\hspace{-0.1cm}\chorz^-=\hor^-_{\text{f}}\\
\hspace{-0.1cm}\chorz^+=\hor^+_{\text{f}}
\end{array}\rd\hspace{-0.2cm},\spc \bs{\Phi}\lp\hor^-_{\text{f}}\rp=\hor^+_{\text{f}} \spc \spc \text{with}\spc\spc H\lp\lambda,y^A\rp=\dfrac{\alpha\lp y^A\rp}{ \ke^+}\big\vert f^-\lp y^A\rp+\ke^-\lambda\big\vert^{\hat{\kappa}}- \dfrac{f^+\lp y^A\rp}{ \ke^+}$,
\item[(b)]$\lb\begin{array}{l}
\hspace{-0.1cm}\chorz^-=\hor^-_{\text{p}}\\
\hspace{-0.1cm}\chorz^+=\hor^+_{\text{p}}
\end{array}\rd\hspace{-0.2cm},\spc \bs{\Phi}\lp\hor^-_{\text{p}}\rp=\hor^+_{\text{p}} \spc \spc \text{with}\spc\spc H\lp\lambda,y^A\rp=-\dfrac{\alpha\lp y^A\rp}{ \ke^+}\big\vert f^-\lp y^A\rp+\ke^-\lambda\big\vert^{\hat{\kappa}}- \dfrac{f^+\lp y^A\rp}{ \ke^+}$,
\end{itemize}
On the other hand, $a<0$ together with \eqref{sign} entail $\text{sign}( f^++\ke^+H)=-\text{sign}( f^-+\ke^-\lambda)$, whence (see (c), (d) in Figure \ref{fig1})
\begin{itemize}
\item[(c)]$\lb\begin{array}{l}
\hspace{-0.1cm}\chorz^-=\hor^-_{\text{f}}\\
\hspace{-0.1cm}\chorz^+=\hor^+_{\text{p}}
\end{array}\rd\hspace{-0.2cm},\spc \bs{\Phi}\lp\hor^-_{\text{f}}\rp=\hor^+_{\text{p}} \spc \spc \text{with}\spc\spc H\lp\lambda,y^A\rp=-\dfrac{\alpha\lp y^A\rp}{ \ke^+}\big\vert f^-\lp y^A\rp+\ke^-\lambda\big\vert^{\hat{\kappa}}- \dfrac{f^+\lp y^A\rp}{ \ke^+}$,
\item[(d)]$\lb\begin{array}{l}
\hspace{-0.1cm}\chorz^-=\hor^-_{\text{p}}\\
\hspace{-0.1cm}\chorz^+=\hor^+_{\text{f}}
\end{array}\rd\hspace{-0.2cm},\spc \bs{\Phi}\lp\hor^-_{\text{p}}\rp=\hor^+_{\text{f}} \spc \spc \text{with}\spc\spc H\lp\lambda,y^A\rp=\dfrac{\alpha\lp y^A\rp}{ \ke^+}\big\vert f^-\lp y^A\rp+\ke^-\lambda\big\vert^{\hat{\kappa}}- \dfrac{f^+\lp y^A\rp}{ \ke^+}$.
\end{itemize}
The function $H$ can be written in a form that covers all cases at once by defining 
\begin{equation}
\epsilon:=\text{sign}\lp\hat{\kappa}\rp\text{sign}(f^-+\ke^-\lambda)
\end{equation}
 and writing
\begin{align}
\label{eqforHd} H\lp\lambda,y^A\rp=\dfrac{1}{\ke^+}\lp \epsilon\alpha\lp y^A\rp\left\vert f^-\lp y^A\rp+\ke^-\lambda\right\vert^{\hat{\kappa}}- f^+\lp y^A\rp\rp,\qquad \alpha\lp y^A\rp>0.
\end{align}
However, it is important to emphasize that the expression of $H$ is only part of the matching,  as the signs also restrict the possible boundaries to be identified as indicated above and in Figure 1. Let us also stress that the matchings (a)-(d) allow for shells with pressure, as the derivatives of \eqref{eqforHd} are given by
\begin{equation}
\lb\begin{array}{l}
\cp_{\lambda}H=\vert a\vert \alpha\lp y^A\rp\vert f^-\lp y^A\rp+\ke^-\lambda\vert^{\hat{\kappa}-1}>0,\\[\medskipamount]
\cp_{\lambda}\cp_{\lambda}H=\dfrac{\lp\hat{\kappa}-1\rp\cp_{\lambda}H}{\frac{f^-\lp y^A\rp}{\ke^-}+\lambda}
\end{array}\rd\quad\Longrightarrow\quad p\lp\lambda,y^A\rp\stackbin{\eqref{finaltau3}}{=}-\dfrac{\hat{\kappa}-1}{\nfi^-\Big( \frac{f^-\lp y^A\rp}{\ke^-}+\lambda\Big)}.
\end{equation}
As discussed in Section 6 in \cite{manzano2021null}, the pressure accounts for the compression/stretching of points when crossing the matching hypersurface. 
This means, in particular, that this effect takes place whenever $\hat{\kappa}\neq1$.

The function $\alpha (y^A)$ introduces a freedom in the matching that we analyse next. Its role is easy to understand when $\mathcal{S}^{\pm}=\emptyset$. Like in the case with vanishing surface gravity of Section \ref{secdeg}, it corresponds to the freedom of selecting a section on each side and impose their identification via $\bs{\Phi}$. The interpretation of this freedom is less obvious when $\mathcal{S}^{\pm}\neq\emptyset$, because these two sections are forced to be mapped to each other. In order to understand this, we again call ``velocity" the rate of change of $s^{\pm}$ along a null generator of $\chorz^{\pm}$. Both when $\chorz^{\pm}$ are degenerate and when $\chorz^{\pm}$ are non-degenerate with $\mathcal{S}^{\pm}=\emptyset$, identifying two sections not only provides a mapping between their points, but also of the velocity along the null generators of $\chor\upm$ at the sections. This information is encoded in the symmetry generators to be identified. However, for non-degenerate $\chorz^{\pm}$ containing fixed point sets $\mathcal{S}^{\pm}\neq\emptyset$, the mapping between the subsets $\mathcal{S}\upm$ only provides information on the identification of their points. The velocity along the null generators remains unfixed, as both symmetry generators vanish on $\mathcal{S}^{\pm}$. The freedom in the function $\alpha$ corresponds precisely to the freedom of selecting the initial velocities at $\mathcal{S}^+$ that rule the identifications off the fixed points set. Once we are off $\mathcal{S}^{\pm}$, the velocity of the identification  is determined by the identification of the symmetry generators themselves, just as in the previous cases. 

\begin{figure}
\centering
\psfrag{M}{$\mathcal{M}^+$}
\psfrag{m}{$\mathcal{M}^-$}
\psfrag{pm}{$\hor_{\text{p}}^-$}
\psfrag{fm}{$\hor_{\text{f}}^-$}
\psfrag{pp}{$\hor_{\text{p}}^+$}
\psfrag{fp}{$\hor_{\text{f}}^+$}
\psfrag{a}{(a)}
\psfrag{b}{(b)}
\psfrag{c}{(c)}
\psfrag{d}{(d)}
\includegraphics[width=16cm]{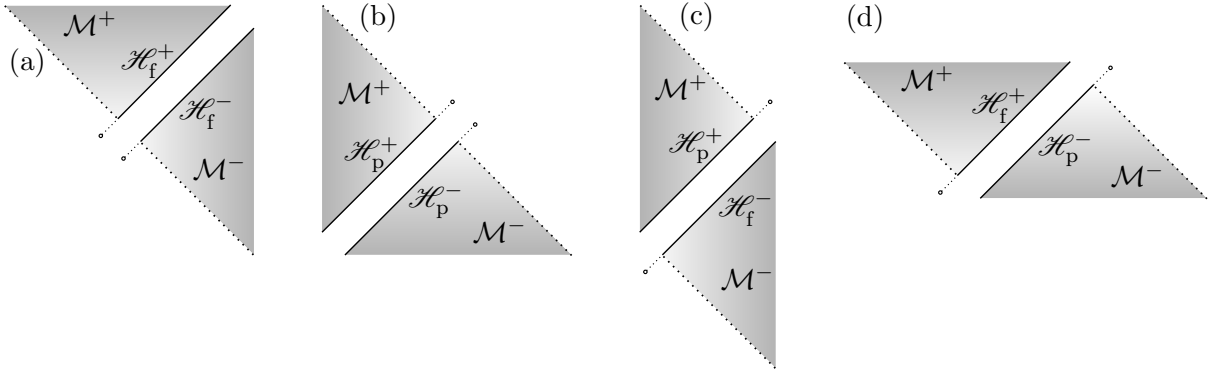} 
\caption{Possible matchings of spacetimes $\lp\Mpm,g^{\pm}\rp$ across their respective boundaries $\chorz^{\pm}$ in the case when they are non-degenerate KH$_0$s without a fixed points set. Here boundaries directly in front of each other are to be identified and the dot represents the point at which the fixed points set would be located if the horizons $\chorz^{\pm}$ extended further. }
\label{fig1}
\end{figure}

\subsection{Case of $\e^-$ degenerate, $\e^+$ non-degenerate}\label{secdegnondeg}
Now we address the case when one boundary is degenerate and the other is not. 
The two possibilities $\ke^-=0$, $\ke^+\neq0$ and $\ke^-\neq0$, $\ke^+=0$ are completely analogous except for the fact that the boundary $\chorz^-$ lies on the future of the spacetime $\lp\Ml,g^-\rp$ while $\chorz^+$ lies on the past of $\lp\Mp,g^+\rp$. 
If the symmetry generator vanishes on the non-degenerate boundary the matching is impossible. This is a direct consequence of Lemma \ref{lemzerosets} because the fixed points set in the degenerate boundary can never be a spacelike cross section. 

We therefore analyze the case when the non-degenerate symmetry generator is everywhere non-vanishing. By Lemma \ref{lemzerosets} again, the symmetry generator of the degenerate boundary must also be free of fixed points, i.e. we have $\mathcal{S}^{\pm} = \emptyset$. Without loss of generality, we take the degenerate symmetry generator to be future. Let first $\horz^{-}$ be the degenerate boundary and note that the choice or causal character of $\xi^-$ requires $f^->0$. Then, \eqref{eqfordH} forces $a( f^++\ke^+H)>0$ and can be integrated to get
\begin{equation}
\label{stepmix1}H\lp\lambda,y^A\rp=\dfrac{1}{a\ke^+}\lp \alpha\lp y^A\rp \exp\lp \frac{a\ke^+\lambda}{f^-\lp y^A\rp}\rp-af^+\lp y^A\rp\rp,\qquad\alpha\lp y^A\rp>0.
\end{equation}
The alternative case when $\horz^{+}$ is the degenerate boundary is analogous. Now $f^+ >0$, $\text{sign}(a)=\text{sign}(f^-+\ke^-\lambda)$ and the integral of \eqref{eqfordH} is
\begin{equation}
\label{stepmix2}H\lp\lambda,y^A\rp=\dfrac{af^+\lp y^A\rp}{\ke^-}\ln\lp\alpha\lp y^A\rp\vert f^-\lp y^A\rp+\ke^-\lambda\vert\rp,\qquad\alpha\lp y^A\rp>0,
\end{equation}
for whatever sign of $a$.  Summarizing, when $a>0$ (resp. $a<0$), a degenerate horizon $\chorz^-$ can be matched with a non-degenerate horizon $\chorz^+\equiv\hor^+_{\text{f}}$ (resp. $\chorz^+\equiv\hor^+_{\text{p}}$) with step function given by \eqref{stepmix1}; and a non-degenerate KH$_0$ $\chorz^-\equiv\hor^-_{\text{f}}$ (resp. $\chorz^-\equiv\hor^-_{\text{p}}$) can be matched with a degenerate horizon $\chorz^+$ with step function \eqref{stepmix2} and $a>0$ (resp. $a<0$). It is worth stressing that the step functions \eqref{stepmix1}-\eqref{stepmix2} are not linear, so the shell has non-zero pressure. Matchings of this type are allowed irrespectively of the extension of the degenerate horizon (which can even be geodesically complete) while the non-degenerate horizon is always limited by the fact that the would-be fixed point set must be absent.  As before, from a physical point of view it is the presence of pressure, and its associated compression/stretching that makes a matching of this type possible. 

We collect the results from Sections \ref{secdeg}, \ref{secnondeg} and \ref{secdegnondeg} in the following theorem. 
\begin{thm}\label{theorem}
In the setup of Sections \ref{secPreliminaries} and \ref{secKH}, let $a\in\mathbb{R}\setminus\{0\}$ and consider the matching of two spacetimes $\lp\Mpm,g\upm\rp$ with boundaries $\chorz^{\pm}$ in which two null symmetry generators $\{\xi^-,a\xi^+\}$ with surface gravities $\{\ke^-,a\ke^+\}$ are to be identified. Then, the fixed points sets $\mathcal{S}^{\pm}$ must be identified via $\bs{\Phi}$ and
\begin{itemize}
\item[(i)] if $\chorz^{\pm}$ are degenerate, the matching is possible with step function \eqref{caseAidKillings};
\item[(ii)] if $\chorz^{\pm}$ are non-degenerate and $\mathcal{S}^{\pm}\neq\emptyset$, the matching requires the surface gravities $\{\ke^-,a\ke^+\}$ to coincide and the step function is given by \eqref{eqforHddef};
\item[(iii)] if $\chorz^{\pm}$ are non-degenerate and $\mathcal{S}^{\pm}=\emptyset$, the only possible matchings are (a)-(d) in Section \ref{secnondeg} with step function \eqref{eqforHd}; \item[(iv)] if $\chorz^-$ (resp. $\chorz^+$) is degenerate with future $\xi^-\vert_{\chor^-}\neq0$ (resp. $\xi^+\vert_{\chor^+}\neq0$) and $\chorz^+$ (resp. $\chorz^-$) is non-degenerate and with $\mathcal{S}^+=\emptyset$ (resp. $\mathcal{S}^-=\emptyset$), the matching can be performed with step function \eqref{stepmix1} (resp. \eqref{stepmix2});
\item[(v)] the matching between a degenerate and a non-degenerate boundaries is impossible when any of them contains fixed points.
\end{itemize}
Moreover, in cases (i) and (ii), the resulting null shell has vanishing pressure. Finally, the matching allows for the freedom of selecting a section on each side and imposing their identification via $\bs{\Phi}$ in (i), (iii) and (iv); and the freedom of setting the initial velocities at $\mathcal{S}^{\pm}$ in (ii).
\end{thm}

\section{Killing Horizons with bifurcation surfaces}\label{secRestrictionY}
From a  physical point of view, perhaps one of the most interesting situations correspond to non-degenerate Killing horizon boundaries with a bifurcation surface. As mentioned in the introduction, this covers all black hole spacetimes with non-zero temperature and whose boundaries are geodesically complete Killing horizons. It therefore makes sense to analyze this case in more detail.  

We already know that a Killing horizon satisfying the assumptions of Sections \ref{secPreliminaries} and \ref{secKH} is, by definition, a KH$_0$. Thus, the matching across non-degenerate Killing horizons $\chor^{\pm}$ containing bifurcation surfaces $\mathcal{S}^{\pm}$ falls into item (ii) in Theorem \ref{theorem}. However, since now we have much stronger conditions, namely the existence of a Killing vector on each side, we can restrict the matching far more. We postpone for a future paper the corresponding analysis for the remaining cases.

Our first task is to find explicit expressions for the tensor fields $Y^{\pm}$ which, as we shall see, turn out to be severely restricted. To ease the notation we drop from now on the $\pm$ in the expressions until the actual matching is performed. 

Consider a non-degenerate boundary $\chor$ with a bifurcation surface $\mathcal{S}$. For this case, a natural choice of coordinates are the so-called R\'acz-Wald coordinates $\{u,v,x^A\}$ \cite{racz1992extensions}. These coordinates can be constructed so that $\chor=\{ u=0 \}$, $\mathcal{S} = \{u=0,v=0\}$ and the non-degenerate Killing vector $\xi$ and the spacetime metric $g$ are given by
\begin{align}
\label{killingRW}\xi&=\kappa_{\xi}(-u\cp_u+v\cp_v),\\
\label{metricRW}g&=-2G(\omega,x^C) dv\lp du+uw_A( \omega,x^C) dx^A\rp+\ovl{\gamma}_{AB}( \omega,x^C) dx^Adx^B,
\end{align}
where $\kappa_{\xi}\in\mathbb{R}$ is the surface gravity of $\xi\vert_{\chor}$, $\omega:=uv$ and $G,w_A,\ovl{\gamma}_{AB}\in\mathcal{F}\lp\mathcal{M}\rp$.
Moreover, these coordinates have a residual freedom that allows one to set $ G\vert_{\chor} = \text{const}$, which we enforce from now on. The non-zero components of the inverse metric are $g^{uv}=-\frac{1}{G}$, $g^{uu}=u^2 \ovl{\gamma}^{AB}w_Aw_B$ and $g^{uA}=-u\ovl{\gamma}^{AB}w_B$, where $\ovl{\gamma}^{AB}$ is the inverse of $\ovl{\gamma}_{AB}$.

As affine null generator of $\chor$, we select\footnote{One immediately checks that $k$ is an affine null generator, as $k^{\alpha}\nabla_{\alpha}k^{\beta}\vert_{\chor}=\Gamma_{vv}^{\beta}\vert_{\chor}=g^{\beta u}\cp_{v}g_{u v}\vert_{\chor}=\delta^{\beta}_v \frac{1}{G}\cp_{v}G\vert_{\chor}=\delta^{\beta}_v \frac{u}{G}\cp_{\omega}G\vert_{\chor}=0$.} $k=\cp_{v}$. The natural choice of scalar function defining a foliation $\{ S_s\}$ of $\chor$ is $s=v$ (recall Definition \ref{def1}). Under these conditions, we construct a basis according to \eqref{basis} by taking $\{ L=\frac{1}{G}\cp_{u},k=\cp_{v},v_I=\cp_{x^I}\}$. The induced metric on the sections $\{ v=\text{const.}\}$ is $h_{IJ}:=\ovl{\gamma}_{IJ}\vert_{\omega=0}$ and we let $\accentset{\circ}{\nabla}$, $\accentset{\circ}{\bs{R}}$ be respectively the corresponding Levi-Civita covariant derivative and Ricci tensor. 

By item \textit{(ii)} in Theorem \ref{theorem}, we know that the surface gravities of the Killing vectors to be identified by the matching must coincide. Moreover, we are allowed to choose a priori these vectors so that they have the same surface gravity on both sides. This clearly entails, in particular, that $a=1$. Observe that the combination of \eqref{killingRW} and \eqref{finaletak} implies $f=0$ which, together with \eqref{eqforHddef}, yields $H(\lambda,y^A)=\alpha(y^A)\lambda$ in this case. 

To determine $Y^{\pm}$ from equations \eqref{Ymenos}-\eqref{YIJ}, it remains to obtain the tensors $\bs{\sigma}_{L}$ and $\bs{\Theta}^{L}$.  
The quantity $\bs{\sigma}_L\lp v_I\rp$ is computed very easily. One gets
\begin{align}
\label{fbfs1}\bs{\sigma}_L\lp v_I\rp\eqch&\spc \frac{1}{G}g\big( \nabla_{\cp_{ x^I }}\cp_v,\cp_u\big)\eqch\frac{1}{G}\Gamma^{\mu}_{v I }g_{\mu u}\eqch -\Gamma^{v}_{v I }\eqch \dfrac{w_I}{2}.
\end{align}
Since $L$ is orthogonal to the leaves $\{v=\text{const.}\}$, $\bs{\sigma}_L$ is the torsion one-form of these sections. Expression \eqref{fbfs1} therefore gives the metric coefficients $w_I\vert_{\chor}$ a geometric meaning.

The computation of $\bs{\Theta}^L\lp v_I,v_J\rp$ is more cumbersome because there appears a term involving derivatives of $\overline{\gamma}$ off $\chor$. To handle this term we use an identity derived in Appendix B. More specifically, we shall use expression \eqref{Riccinondeg2}, which relates $\partial_{\omega} \overline{\gamma}_{IJ}$ with geometric objects on the boundary and with the spatial tangent components of the ambient Ricci tensor, namely $R_{IJ} := \textbf{Ric}(v_A,v_B)$. We compute
\begin{align}
 \nonumber\bs{\Theta}^L\lp v_I,v_J\rp\eqch&\spc \frac{1}{G}g\big( \nabla_{\cp_{x^I }}\cp_u,\cp_{ x^J }\big)\eqch\frac{1}{G}\Gamma^{\mu}_{u I }g_{\mu J }\eqch\frac{1}{G}\Gamma^{A}_{u I }\ovl{\gamma}_{A J }\eqch\dfrac{v}{2G}\cp_{\omega}\ovl{\gamma}_{IJ}\\
 \label{fbfs2}\stackbin[\eqref{Riccinondeg2}]{\chor}{=}&\spc\dfrac{v}{2}\lp R_{IJ}-\accentset{\circ}{R}_{IJ}-\dfrac{1}{2}\lp \accentset{\circ}{\nabla}_Iw_J+\accentset{\circ}{\nabla}_Jw_I\rp+\dfrac{1}{2}w_Iw_J\rp.
\end{align}
Observe that in the present case $\bs{\Theta}^L$ is the second fundamental form of the sections $\{v=\text{const}.\}$ in the direction $L$, again due to the fact that $L$ is perpendicular to the foliation. 

As also proven in Appendix \ref{appedixA}, each term inside the parenthesis of \eqref{fbfs2} is independent of the coordinate $v$. In particular, this applies to $\accentset{\circ}{R}_{AB}^{\pm}$ which, in addition, is forced to verify that 
\begin{equation}
\label{RicciCondpreliminary}\accentset{\circ}{R}^-_{AB}\vert_{p}=b_A^Ib_B^J\accentset{\circ}{R}_{IJ}^+\vert_{\bs{\Phi}(p)},\qquad\forall p\in\mathcal{S}^-,\quad \bs{\Phi}(p)\in\mathcal{S}^+
\end{equation}
because the bifurcation surfaces $\mathcal{S}^{\pm}$ must be isometric and mapped to each other via $\bs{\Phi}$. The scalars $b_I^J$ on $\chor^+$ do not depend on $v$ either, because from \eqref{evectors2}-\eqref{b_I^J} 
\begin{equation}
k^+(b_I^J)=\frac{e^+_1(b_I^J)}{\cp_{\lambda}H}=\frac{\cp_{\lambda}b_I^J}{\cp_{\lambda}H}=0.
\end{equation}
Consequently, $\accentset{\circ}{R}^-_{AB}$ and $b_A^Ib_B^J\accentset{\circ}{R}^-_{IJ}$ take the same value for all points of the null generators containing $p\in\mathcal{S}^-$ and $\bs{\Phi}(p)\in\mathcal{S}^+$ respectively. The fact that null generators must be identified by the matching entails that condition \eqref{RicciCond} holds everywhere, i.e. 
\begin{equation}
\label{RicciCond}\accentset{\circ}{R}^-_{AB}\vert_{p}=b_A^Ib_B^J\accentset{\circ}{R}_{IJ}^+\vert_{\bs{\Phi}(p)},\qquad\forall p\in\chor^-,\quad \bs{\Phi}(p)\in\chor^+.
\end{equation}
From now on we remove the explicit writing of $p$ and $\bs{\Phi}(p)$ in this expression and similar ones. The trivial identification between $\chor^-$ and $\Sigma$ ensures that the pullback $(\Phi^-)^*(\accentset{\circ}{\bs{R}}{}^-)$ coincides with the Ricci tensor $\bs{R}^{\parallel}$ on the sections $\{\lambda=\text{const.}\}\subset\Sigma$. Consequently, it must hold that $\bs{R}^{\parallel}=(\Phi^{\pm})^*(\accentset{\circ}{\bs{R}}{}^{\pm})$ (cf. \eqref{tensorremark1}-\eqref{tensorremark2}). 

In \eqref{fbfs2} there appear covariant derivatives of the one-forms $\bs{w}^{\pm}:=w^{\pm}_Idx^I_{\pm}$ on $\chor^{\pm}$. We need to compute their pullbacks onto $\Sigma$ and relate the result to $\nabla^{\parallel}$ derivatives of the corresponding pullback one-forms $\nwpm :=(\Phi^{\pm})^*(\bs{w}^{\pm})$. In Appendix \ref{apppullbacks} we derive a general identity of this type for totally geodesic null boundaries, valid for general covariant tensors that annihilate the
null generators on each of its entries. Applying Lemma \ref{lempullback} in Appendix \ref{apppullbacks} together with $k^{\pm}(w_I^{\pm}) =0$ it follows
\begin{equation}
w_I^-=\Wm_I,\qquad b_I^Jw^+_J=\Wp_I,\qquad \accentset{\circ}{\nabla}^-_Iw^-_J=\nabla^{\parallel}_I\Wm_J,\qquad b_I^Ab_J^B\accentset{\circ}{\nabla}^+_Aw^+_B=\nabla^{\parallel}_I\Wp_J.
\end{equation}
Inserting this into \eqref{fbfs1} and \eqref{fbfs2} we obtain that the pullback tensors $\nsigmapm$ and $\nThetapm$ become
\begin{align}
\label{sigmapmND}\sigmapm_I&=\dfrac{\Wpm_I}{2},\\
\Thetam_{IJ}&=\dfrac{\lambda}{2}\lp \Rm_{IJ}-{R}^{\parallel}_{IJ}-\dfrac{1}{2}\lp \nabla^{\parallel}_I\Wm_J+\nabla^{\parallel}_J\Wm_I\rp+\dfrac{1}{2}\Wm_I\Wm_J\rp,\\
\label{ThetapmND}\Thetap_{IJ}&=\dfrac{\alpha\lambda}{2}\lp \Rp_{IJ}-{R}^{\parallel}_{IJ}-\dfrac{1}{2}\lp \nabla^{\parallel}_I\Wp_J+\nabla^{\parallel}_J\Wp_I\rp+\dfrac{1}{2}\Wp_I\Wp_J\rp.
\end{align}
It is useful to introduce the following tensors $\nvsigmapm$ and $\nXipm$ on the leaves $\{\lambda=\text{const.}\}\subset\Sigma$
\begin{equation}
\label{sigmaandXi}\begin{array}{l}
\vsigmam_I:=\Wm_I,\\
\vsigmap_I:=\Wp_I-2\dfrac{\nabla_{I}^{\parallel}\alpha}{\alpha},
\end{array}\qquad \Xipm_{IJ}:=\dfrac{1}{2}\lp \Rpm_{IJ}-{R}^{\parallel}_{IJ}-\dfrac{1}{2}\lp \nabla^{\parallel}_I\vsigmapm_J+\nabla^{\parallel}_J\vsigmapm_I\rp+\dfrac{1}{2}\vsigmapm_I\vsigmapm_J\rp.
\end{equation}
To understand why these tensors are of relevance, let us relate
$\nXip$ and $\nThetap$. For that we compute $\nabla^{\parallel}_I\vsigmap_J+\nabla^{\parallel}_J\vsigmap_I$ and $\vsigmap_I\vsigmap_J$ to get
\begin{align}
\nabla^{\parallel}_I\vsigmap_J+\nabla^{\parallel}_J\vsigmap_I&=\nabla^{\parallel}_I\Wp_J+\nabla^{\parallel}_J\Wp_I-4\dfrac{\nabla^{\parallel}_I\nabla^{\parallel}_J\alpha}{\alpha}+4\dfrac{(\nabla^{\parallel}_I\alpha)(\nabla^{\parallel}_J\alpha)}{\alpha^2},\\
\vsigmap_I\vsigmap_J&=\Wp_I\Wp_J-\dfrac{2}{\alpha}\lp \Wp_I\nabla^{\parallel}_J\alpha+\Wp_J\nabla^{\parallel}_I\alpha\rp+4\dfrac{(\nabla^{\parallel}_I\alpha)(\nabla^{\parallel}_J\alpha)}{\alpha^2}.
\end{align}
Comparing with \eqref{ThetapmND} we conclude that
\begin{align}
\label{Xiplus}\Xip_{IJ}\lambda=&\spc\lp\dfrac{\nabla^{\parallel}_I\nabla^{\parallel}_J\alpha}{\alpha}-\dfrac{1}{2\alpha}\lp \Wp_I\nabla^{\parallel}_J\alpha+\Wp_J\nabla^{\parallel}_I\alpha\rp\rp\lambda+\dfrac{\Thetap_{IJ}}{\alpha}.
\end{align}
This relation simplifies drastically the final expression for the tensors $Y^{+}$ and $\tau$ as given in Proposition \ref{propenergymomtensor}. Particularizing \eqref{Ymenos}-\eqref{finaltau3} to $\nfi=1$, $\widetilde{\psi}{}_I^{\pm}=0$, $H(\lambda,y^A)=\alpha(y^A)\lambda$ and $\widetilde{\bs{\chi}}=0$, one finds
\begin{align}
\label{Yrestricted-}
&\begin{array}{lll}
Y^{-}_{11} = 0, & Y^{-}_{1J} =-\sigmam_J, & Y^{-}_{IJ} = \Thetam_{IJ},\\[\medskipamount]
Y_{11}^+= 0, & Y_{1J}^+= -\sigmap_J+\dfrac{\cp_{y^J}\alpha}{\alpha}, & Y^+_{IJ}=\dfrac{1}{\alpha}\lp\lp\nabla_{I}^{\parallel}\nabla_{J}^{\parallel}\alpha-2 \nabla_{{\lp I\rd}}^{\parallel}\alpha\spc \sigmap_{\ld J\rp}\rp\lambda +\Thetap_{IJ}\rp,
\end{array}\\
\label{taurestricted1}&\spc\tau^{IJ}=0,\qquad \tau^{1I}=\gamma^{IJ}\lp \dfrac{\cp_{y^J}\alpha}{\alpha}-[\widetilde{\sigma}_J]\rp ,\\
\label{taurestricted2}&\spc\tau^{11}=-\gamma^{IJ}\lp\dfrac{1}{\alpha}\lp\nabla_{I}^{\parallel}\nabla_{J}^{\parallel}\alpha-2 \nabla_{{\lp I\rd}}^{\parallel}\alpha\spc \sigmap_{\ld J\rp}\rp\lambda +\dfrac{\Thetap_{IJ}}{\alpha}-\Thetam_{IJ}\rp.
\end{align}
Inserting \eqref{sigmapmND}-\eqref{sigmaandXi} and \eqref{Xiplus} into \eqref{Yrestricted-}-\eqref{taurestricted2}, we get the final expressions collected in the following theorem.
\begin{thm}\label{theoremNDYtau}
Assume the setup of Theorem \ref{theorem} and consider a matching in which the boundaries $\chor^{\pm}$ are non-degenerate Killing horizons containing bifurcation surfaces $\mathcal{S}^{\pm}$. Construct R\'acz-Wald coordinates $\{u_{\pm},v_{\pm},x^A_{\pm}\}$ so that $\chor^{\pm}=\{u_{\pm}=0\}$ and $\mathcal{S}^{\pm}=\{u_{\pm}=0,v_{\pm}=0\}$ and in which the Killing vector fields and the corresponding metrics are given by \eqref{killingRW}-\eqref{metricRW}. Select $\xi^{\pm}$ to have the same surface gravities (which forces $a=1$) and let $\alpha=\cp_{\lambda}H$. Then, the tensors $Y^{\pm}$ and the energy-momentum tensor of the shell $\tau$ can be expressed in terms of the tensors
\begin{equation}
\label{defsTheorem2}
\begin{array}{l}
\vsigmam_I:=\Wm_I,\\
\vsigmap_I:=\Wp_I-2\dfrac{\nabla_{I}^{\parallel}\alpha}{\alpha},
\end{array}\qquad \Xipm_{IJ}:=\dfrac{1}{2}\lp \Rpm_{IJ}-{R}^{\parallel}_{IJ}-\dfrac{1}{2}\lp \nabla^{\parallel}_I\vsigmapm_J+\nabla^{\parallel}_J\vsigmapm_I\rp+\dfrac{1}{2}\vsigmapm_I\vsigmapm_J\rp,
\end{equation}
as
\begin{align}
\label{YNDfinal}Y^{-}_{11} &=0, \quad Y^{-}_{1J} =-\dfrac{\vsigmam_J}{2},   & Y^{-}_{IJ} &= \Xim_{IJ}\lambda; &  \quad Y_{11}^+&=0, \quad Y_{1J}^+= -\dfrac{\vsigmap_J}{2},\quad Y^+_{IJ}=\Xip_{IJ}\lambda;\\
\label{tauNDRestriction}\tau^{11}&=-\gamma^{IJ} [\widetilde{\Xi}_{IJ}] \lambda, & \tau^{1I}&=-\dfrac{1}{2}\gamma^{IJ}[\widetilde{\varsigma}_J],& \tau^{IJ}&=0.
\end{align}
\end{thm}
\begin{rem}
Note the intrinsic curvature term $\bs{R}^{\parallel}$ drops out from the jump $[\widetilde{\bs{\Xi}}]$. The underlying reason is the already mentioned isometry condition $(\Phi^+)^*(\accentset{\circ}{\bs{R}}^+)=(\Phi^-)^*(\accentset{\circ}{\bs{R}}^-)$. 
\end{rem}
The results \eqref{YNDfinal}-\eqref{tauNDRestriction} allow us to conclude that the matter-content of the shell, given by $ Y^{\pm}$ and $\tau$, exclusively depends on the choice of $\alpha$, on the intrinsic and extrinsic geometry of the bifurcation surfaces $\mathcal{S}^{\pm}$ (recall that $\bs{\sigma}_L$ and $\bs{\Theta}^L$ are the torsion one-form and the transverse null second fundamental form in this case) and on the pullback to these surfaces of the Ricci tensors of the ambient spacetimes $\lp\Mpm,g^{\pm}\rp$, which together determine the form of the tensors $\nvsigmapm$ and $\nXipm$ according to the equations \eqref{fbfs1}, \eqref{sigmaandXi} and \eqref{Xiplus}. 

The components of $Y^{\pm}$, $\tau$ are either constant along the null generators or linear in $\lambda$. It is worth mentioning that the energy density of the shell is either identically zero or unavoidably changes its sign at the bifurcation surface. Besides, the energy current $j^I$ is independent of $\lambda$, which means that the flux of energy is insensitive to the change of sign on the energy of the shell. This raises some questions concerning the physical interpretation of the quantities $\rho:=\tau^{11}$, $j^I:=\tau^{1I}$ and $p:=(n-1)^{-1}\gamma_{AB}\tau^{AB}$. 
We include below some comments in this regard.

Let us call velocity the rate of change of the foliation defining functions along the null generators $e^{\pm}_1$ and acceleration to the rate of change of the velocity. As discussed in \cite{manzano2021null}, the pressure $p$ accounts for the effect of self-compression or self-stretching of points when crossing from $\chor^-$ to $\chor^+$. The identification between $\chor^-$ and $\Sigma$ always gives velocity equal to one on this side. For this reason, the effect of self-compression/self-stretching only appears when there exists non-constant acceleration along the generators of $\chor^+$. As also shown in  \cite{manzano2021null}, the energy density of the shell increases when points are compressed and vice versa.

Nevertheless, despite the shell has vanishing pressure in the present case, some effect of compression or stretching of points is still taking place because the velocity along the null generators of $\chor^+$, ruled by the function $\alpha$, is different for each generator. As a consequence there appears an energy flux which points toward null generators with higher values of $\alpha$, i.e. with greater velocities (which, however, are still constant along the generator, hence yielding zero pressure). Note that the fact that the currents depend on the jump of the torsion one-forms could also be expected, as these tensors are the projection on $k$ of the variation $\nabla^{\pm}{}_{v^{\pm}_I}k^{\pm}$ along the leaves $\{s^{\pm}=\text{const.}\}\subset\chor^{\pm}$. The greater the jump of these tensors, the more difference there is in the lengths of $k^{\pm}$, which yields a greater jump on the velocities along the generators of both sides.

We find the change of sign on the energy density of the shell $\rho$ across the bifurcation surface really puzzling, and we do not know yet how to interpret this. The result suggests that the causality change of the Killing fields from future to past across the bifurcation surface somehow affects the energy density of the shell. We emphasize however, that this behaviour is fully compatible with the shell field equations obtained by Barrabés and Israel \cite{barrabes1991thin} for the case of null hypersurfaces. This of course had to be the case and we include an explicit proof in next section because this yields a non-trivial consistency check of our results.

\subsection{Surface layer equations}
The tensors $Y^{\pm}$ and energy-momentum tensor on a shell satisfy the so-called Israel equations (also known as {\it shell equations} or {\it surface layer equations}). In the framework of hypersurface data, they read \cite{mars2013constraint}
\begin{align}
\label{sfe1}\dfrac{1}{\sqrt{\vert\det\bs{\mathcal{A}}\vert}}\cp_{\theta^a}\lp \sqrt{\vert\det\bs{\mathcal{A}}\vert}\tau^{ab}\ell_b\rp-\dfrac{1}{2}\tau^{ab}\lp Y_{ab}^++Y_{ab}^-\rp=&\lc\rho_{\ell}\rc,\\
\label{sfe2}\dfrac{1}{\sqrt{\vert\det \bs{\mathcal{A}}\vert}}\cp_{\theta^b}\lp \sqrt{\vert\det\bs{\mathcal{A}}\vert}\tau^{bc}\gamma_{ca}\rp-\dfrac{1}{2}\tau^{bd}\cp_{\theta^a}\gamma_{bd}=&\lc J_a\rc,
\end{align}
where $\{\theta^1=\lambda,\theta^I=y^I\}$, $\lc\rho_{\ell}\rc:=\rho_{\ell}^+-\rho_{\ell}^-$, $\lc J_a\rc:=J_a^+-J_a^-$, and the bulk energy and momentum quantities $\rho^{\pm}_{\ell}$, $J_a^{\pm}$ are defined by 
\begin{align}
\label{Jrho}\rho_{\ell}^{\pm}:=-({\Phi^{\pm}})^*\lp \textbf{Ein}^{\pm}\lp\zeta^{\pm},\nu^{\pm}\rp\rp ,\qquad \bs{J}^{\pm}:=-({\Phi^{\pm}})^*\lp \textbf{Ein}^{\pm}\lp\cdot,\nu^{\pm}\rp\rp.
\end{align}
Here $\textbf{Ein}^{\pm}$ denotes the Einstein tensor of $\lp\Mpm,g^{\pm}\rp$ and $\nu^{\pm}$ is the (unique) normal vector to the hypersurface, normalized to $g^{\pm}( \zeta^{\pm}, \nu^{\pm} )=1$. 

In the present case $\nu=-e_1^{\pm}$. Moreover, $\textbf{Ein}^{\pm}(k^{\pm},k^{\pm})=0$ as a consequence of the Raychaudhuri equation  and $\textbf{Ein}^{\pm}(k^{\pm},v_I^{\pm})=0$ because the surface gravities $\ke^{\pm}$ are constant on $\chor^{\pm}$ (see e.g. equations (9.2.11), (12.5.22) and (12.5.30) in \cite{wald1984general}). This immediately entails that $\bs{J}^{\pm}=0$. On the other hand, in terms of the basis $\{L^{\pm},k^{\pm},v_A^{\pm}\}$ we can decompose the inverse metric as $
g_{\pm}^{\alpha\beta}=h_{\pm}^{AB}(v_A^{\pm})^{\alpha}(v_B^{\pm})^{\beta}-(L^{\pm})^{\alpha}(k^{\pm})^{\beta}-(L^{\pm})^{\beta}(k^{\pm})^{\alpha}$. This yields Ricci scalar $\bs{R}^{\pm}\vert_{\chor^{\pm}}=h_{\pm}^{AB}\textbf{Ric}^{\pm}(v_A^{\pm}, v_B^{\pm})-2\textbf{Ric}^{\pm}(L^{\pm},k^{\pm})$ and hence
\begin{align}
\label{SEeq23}\rho_{\ell}^{\pm}\eqchpm\frac{1}{2}h_{\pm}^{AB}\textbf{Ric}^{\pm}(v_A^{\pm}, v_B^{\pm})\qquad\Longrightarrow\qquad \rho_{\ell}^{\pm}=\frac{1}{2}\gamma^{IJ}\widetilde{R}_{IJ}^{\pm},
\end{align}
where for the implication we made use of the isometry condition \eqref{isomcondpaper1}. 

To prove that the shell equations hold for the tensor fields $Y^{\pm}$ and $\tau$ of Theorem \ref{theoremNDYtau}, we compute each term of the left hand side of \eqref{sfe1}-\eqref{sfe2} separately. We start with \eqref{sfe1}. The tensor $\bs{\mathcal{A}}$ (cf. \eqref{ambientmetric}) is given in this case by 
\begin{equation}
\bs{\mathcal{A}}=\lp
\begin{array}{ccc}
0 & 0 & -1\\
0 & \gamma & 0\\
-1 & 0 & 0\\
\end{array}
\rp,
\end{equation}
because $\ell_1=-\nfi^-=-1$ and $\ell_I=-\psi_I^-=0$. Consequently, $\vert\det\gamma\vert=\vert\det\bs{\mathcal{A}}\vert$. On the other hand, the fact that $\chor^{\pm}$ are totally geodesic entails $k^{\pm}(h^{\pm}_{IJ})=0$ (see \eqref{k(hIJ)}), from where it follows that $0=e^-_{1}(h^-_{IJ})=\cp_{\lambda}\gamma_{IJ}$. For spatial derivatives of $\vert\det\gamma\vert$, we use the well-known identity
\begin{equation}
\label{SEeq0}\dfrac{1}{\vert\det\gamma\vert}\cp_{y^I}(\vert\det\gamma\vert)={\Gamma^{\parallel}}{}^{A}_{AI},
\end{equation}
where $\Gamma^{\parallel}{}^{A}_{BI}$ are the Christoffel symbols of $\nabla^{\parallel}$. The following expressions are immediate consequences of \eqref{YNDfinal}-\eqref{tauNDRestriction} together with $l_a = -\delta_a^1$ and
$\gamma_{1a} =0$:
\begin{align}
\label{SEeq1}\tau^{ab}\ell_b&=-\tau^{1a}\\
\label{SEeq2}\tau^{ab}Y_{ab}^{\pm}&=\tau^{11}Y_{11}^{\pm}+2\tau^{1I}Y_{1I}^{\pm}+\tau^{IJ}Y_{IJ}^{\pm}=2\tau^{1I}Y_{1I}^{\pm}\\
\label{SEeq8}\tau^{ab}(Y_{ab}^++Y^-_{ab})&=\dfrac{1}{2}\gamma^{IJ}[\widetilde{\varsigma}_J](\vsigmap_I+\vsigmam_I)=\dfrac{1}{2}\gamma^{IJ}[\vsigma_I\vsigma_J].\\
\label{SEeq3}\tau^{bc}\gamma_{ca}&=\tau^{bJ}\gamma_{aJ}=\delta_{1}^b\delta_a^I\tau^{1J}\gamma_{IJ}\\
\label{SEeq4}\tau^{bd}\cp_{\theta^a}\gamma_{bd}&=\tau^{11}\cp_{\theta^a}\gamma_{11}+2\tau^{1I}\cp_{\theta^a}\gamma_{1I}+\tau^{IJ}\cp_{\theta^a}\gamma_{IJ}=0.
\end{align}
By \eqref{SEeq0} and \eqref{SEeq1}, it follows
\begin{align}
\nonumber \dfrac{1}{\sqrt{\vert\det\bs{\mathcal{A}}\vert}}\cp_{\theta^a}\lp \sqrt{\vert\det\bs{\mathcal{A}}\vert}\tau^{ab}\ell_b\rp&=-\dfrac{1}{\sqrt{\vert\det\gamma\vert}}\cp_{\theta^a}(\sqrt{\vert\det\gamma\vert}\tau^{1a})\\
\nonumber &=-\cp_{\lambda}\tau^{11}-\dfrac{1}{\sqrt{\vert\det\gamma\vert}}\cp_{y^I}( \sqrt{\vert\det\gamma\vert}\tau^{1I})\\
\nonumber &=\gamma^{IJ}[\widetilde{\Xi}_{IJ}]-\nabla^{\parallel}_I\tau^{1I}\\
\nonumber &=\gamma^{IJ}[\widetilde{\Xi}_{IJ}]+\dfrac{1}{4}\gamma^{IJ}(\nabla^{\parallel}_I[\widetilde{\varsigma}_J]+\nabla^{\parallel}_J[\widetilde{\varsigma}_I]),\\
\label{SEeq21} &=\dfrac{1}{2}\gamma^{IJ}[\widetilde{R}_{IJ}]+\dfrac{1}{4}\gamma^{IJ}[\vsigma_I\vsigma_J],
\end{align}
where in the last equality we inserted \eqref{defsTheorem2} and \eqref{tauNDRestriction}. Combining \eqref{SEeq23}, \eqref{SEeq8} and \eqref{SEeq21}, the shell equation \eqref{sfe1} follows immediately.

Checking the validity of equation \eqref{sfe2} is almost direct. Since $\bs{J}^{\pm}$ are zero, it suffices to substitute \eqref{SEeq3}-\eqref{SEeq4} into \eqref{sfe2} to obtain
\begin{equation}
0=\dfrac{1}{\sqrt{\vert\det \gamma\vert}}\cp_{\theta^b}\lp \sqrt{\vert\det\gamma\vert}\delta_{1}^b\delta_a^I\tau^{1J}\gamma_{IJ}\rp=\dfrac{1}{\sqrt{\vert\det \gamma\vert}}\cp_{\lambda}\lp \sqrt{\vert\det\gamma\vert}\delta_a^I\tau^{1J}\gamma_{IJ}\rp,
\end{equation}
which is automatically satisfied as nothing inside the parenthesis depends on $\lambda$.

\section{Spherical, plane or hyperbolic symmetric spacetimes} \label{secNDexamples}
To conclude this paper, we apply the formalism above to study particular matchings of interest. We start by determining the necessary and sufficient conditions that allow for the matching of two arbitrary spherical, plane or hyperbolic symmetric spacetimes admitting a Killing horizon with a bifurcation surface. We then particularize the results for the cases of two Schwarzschild spacetimes and two Schwarzschild-de Sitter spacetimes. As usual, we avoid $\pm$ notation until the actual matching is performed. 

Let $\lp\mathcal{M},g\rp$ be a spherical, plane or hyperbolic symmetric spacetime and $\Lambda$ be its corresponding cosmological constant. Assume that it admits a Killing vector field $\xi$ defining a  bifurcation surface $\mathcal{S}\subset\mathcal{M}$. Any spacetime of this kind is by definition a warped product of a $2$-dimensional Lorentzian manifold $\lp\mathcal{N},\bar{g}\rp$ and an $\lp n-1\rp$-dimensional Riemannian space $\lp\mathcal{W}, h_{\vk}\rp$ of constant curvature $\varkappa\in\{ 1,0,-1\}$ \cite{chen2011pseudo}, \cite{neil1983semi}. We let $r\in\mathcal{F}\lp\mathcal{N}\rp$ be the warping function and use R\'acz-Wald coordinates $\{ u,v,x^A\}$ constructed as in Section \ref{secRestrictionY}. We again scale a priori the Killing vectors defining the horizons of each spacetime  so that they have the same surface gravity. In terms of $\omega:=uv$, the warped metric is 
\begin{equation}
\label{warpedmetric}
g=\bar{g}+r^2\lp\omega\rp h_{\vk}\lp x^A\rp,
\end{equation}
where $\bar{g}=-2G\lp\omega\rp du dv$, $G\in\mathcal{F}^*\lp\mathcal{M}\rp$ (note that $G\vert_{\chor}$ is constant). 

The induced metric on the bifurcation surfaces $\mathcal{S}^{\pm}=\{u_{\pm}=0,v_{\pm}=0\}$ is $g^{\pm}\vert_{\mathcal{S}^{\pm}}=r_{\pm}^2h_{\vk}^{\pm}\vert_{\mathcal{S}^{\pm}}=(r^{\pm}_0)^2h_{\vk}^{\pm}$, where $r_0:= r\vert_{\chor} \neq 0$. The map $\bs{\Phi}:\mathcal{S}^-\longrightarrow\mathcal{S}^+$ must be an isometry so the Ricci scalars of $g^{\pm}\vert_{\mathcal{S}^{\pm}}$, which in this case are given by  
\begin{equation}
( n-1)( n-2)\varkappa^{\pm}( r_0^{\pm})^{-2},
\end{equation}
must be preserved by $\bs{\Phi}$. Therefore
\begin{equation}
\label{matchcondsym}\dfrac{\varkappa^-}{(r_0^-)^{2}}=\dfrac{\varkappa^+}{( r_0^+)^{2}},
\end{equation}
and we conclude that $\varkappa^{\pm}$ must coincide (recall that $\varkappa^{\pm}\in\{ 1,0,-1\}$). From now on we simplify the notation and write $\varkappa$ instead of $\varkappa^{\pm}$. An immediate consequence of \eqref{matchcondsym} is that the jump $\lc r_0\rc:=r^+_0-r^-_0$ is zero whenever $\varkappa\neq0$, and can take whatever value in the plane case with $\varkappa=0$.

Since $\chor^{\pm}$ are totally geodesic, equation \eqref{isomcondpaper1} constitutes an isometry condition between the leaves $\{s^{\pm}=\text{const.}\}\subset\chor^{\pm}$. These sections are euclidean planes, spheres of radius $r_0$ and hyperbolic planes when $\varkappa=0$, $\varkappa=1$ and $\varkappa=-1$ respectively. The corresponding isometries are respectively euclidean motions, rotations and hyperbolic rotations. In each case they are also isometries of the ambient spacetimes, so the freedom in the matching, encoded in $\bs{\Phi}$, can be absorbed (with full generality) in the coordinates $\{u_+,v_+,x_+^A\}$ in such a way that the coefficients $b_I^J$ take the simple form $b_I^J=\delta_I^J$. This will be assumed from now on. Thus (cf. \eqref{emhd}, \eqref{evectors1}-\eqref{evectors2})
\begin{equation}
\label{gammaandhsym}\gamma_{IJ}:=g^{\pm}\lp e^{\pm}_I,e^{\pm}_J\rp=(r_0^{\pm})^2h_{\vk}^{\pm}{}_{IJ}.
\end{equation}

The metric \eqref{warpedmetric} is of the form \eqref{metricRW} with $w^{\pm}_A=0$ and $\ovl{\gamma}^{\pm}=r_{\pm}^2h_{\vk}^{\pm}$. Consequently, the torsion one-forms $\bs{\sigma}^{\pm}_{L^{\pm}}$ vanish on $\chor^{\pm}$ (cf. \eqref{fbfs1}). Inserting this into \eqref{defsTheorem2}, it follows that 
\begin{equation}
\label{eqnumberrandom}\hspace*{-0.2cm}\vsigmam_I=0,\quad\vsigmap_I=-2\dfrac{\nabla^{\parallel}_I\alpha}{\alpha},\quad\Xim_{IJ}=\dfrac{1}{2}\lp \widetilde{R}{}^{-}_{IJ}-R^{\parallel}_{IJ}\rp,\quad\Xip_{IJ}=\dfrac{1}{2}\lp \Rp_{IJ}-{R}^{\parallel}_{IJ}+2\dfrac{\nabla^{\parallel}_I\nabla^{\parallel}_J\alpha}{\alpha}\rp.
\end{equation}
The fact that the metric \eqref{gammaandhsym} is of constant curvature $\vk   r_0^{-2}$ means that its Ricci tensor is 
\begin{equation}
\label{Xitensor}R^{\parallel}_{IJ}=\dfrac{(n-2)\vk}{(r_0^{\pm})^2}\gamma_{IJ},\quad\text{ and therefore }\quad \lb \begin{array}{l}
\Xim_{IJ}=\frac{1}{2}\lp \widetilde{R}{}^{-}_{IJ}-\frac{(n-2)\vk}{(r_0^{-})^2}\gamma_{IJ}\rp,\\ [\medskipamount]
\Xip_{IJ}=\frac{1}{2}\Big( \Rp_{IJ}-\frac{(n-2)\vk}{(r_0^{+})^2}\gamma_{IJ}+2\frac{\nabla^{\parallel}_I\nabla^{\parallel}_J\alpha}{\alpha}\Big).
\end{array} \rd
\end{equation}
Substituting this in the expressions of Theorem \ref{theoremNDYtau} and using \eqref{matchcondsym}, we obtain
\begin{align}
\label{Ysym}&\lb\begin{array}{l}
Y^{-}_{11}= 0,\\[\medskipamount]
Y^{-}_{1J} =0,\\[\medskipamount]
Y^{-}_{IJ} = \frac{\lambda}{2}\big( \widetilde{R}^{-}_{IJ}-\frac{(n-2)\vk}{(r_0^{-})^2}\gamma_{IJ}\big),
\end{array}\rd\qquad\lb\begin{array}{l}
Y_{11}^+= 0,\\
Y_{1J}^+= \frac{\cp_{y^J}\alpha}{\alpha},\\
Y^+_{IJ}=\Big(\frac{\nabla_{I}^{\parallel}\nabla_{J}^{\parallel}\alpha}{\alpha} +\frac{1}{2}\big( \widetilde{R}^{+}_{IJ}-\frac{(n-2)\vk}{(r_0^{+})^2}\gamma_{IJ}\big)\Big)\lambda,
\end{array}\rd\\
\label{tausym}&\spc\spc\quad\tau^{IJ}=0,\qquad \tau^{1I}=\gamma^{IJ} \frac{\cp_{y^J}\alpha}{\alpha} ,\qquad \tau^{11}=-\gamma^{IJ}\lp \frac{\nabla_{I}^{\parallel}\nabla_{J}^{\parallel}\alpha}{\alpha} +\frac{1}{2}[\widetilde{R}_{IJ}]\rp\lambda.
\end{align}
The resulting shells have therefore energy density $\rho$ and energy flux $j^J$ given by
\begin{equation}
\rho=-\gamma^{IJ}\lp \frac{\nabla_{I}^{\parallel}\nabla_{J}^{\parallel}\alpha}{\alpha} +\frac{1}{2}[\widetilde{R}_{IJ}]\rp\lambda,\quad\qquad j^J=\gamma^{IJ}\frac{\cp_{y^J}\alpha}{\alpha}.
\end{equation} 
An interesting particular case is when the spacetimes $\lp\Mpm,g^{\pm}\rp$ to be matched are, in addition, solutions to the $\Lambda^{\pm}$-vacuum Einstein field equations
\begin{equation}
 R^{\pm}_{\alpha\beta}=\frac{2\Lambda^{\pm}}{n-1}g^{\pm}_{\alpha\beta},
\end{equation}
which impose
\begin{equation}
\label{EFEsymmetric} R^{\pm}_{AB}\stackbin{\chor^{\pm}}{=}\frac{2\Lambda^{\pm}}{n-1}(r_0^{\pm})^2h_{\vk}{}_{AB}\stackbin{\chor^{\pm}}{=}\frac{2\Lambda^{\pm}}{n-1}\gamma_{AB}\quad\Longrightarrow\quad \Rpm_{AB}:=(\Phi^{\pm})^*(\bs{R}^{\pm})_{AB}=\frac{2\Lambda^{\pm}}{n-1}\gamma_{AB}.
\end{equation}
Inserting this into \eqref{Ysym}-\eqref{tausym} yields
\begin{equation}
\label{Ytauwarped}
\hspace{-0.2cm}\lb\begin{array}{ll}
Y^-_{11}=0, &  Y^+_{11}=0,\\
Y^-_{1J}=0, & Y^+_{1J}=\frac{\cp_{y^J}\alpha}{\alpha},\\
Y^-_{IJ}=\beta^-\lambda \gamma_{IJ} , &  Y^+_{IJ}=\big(\frac{\nabla^{\parallel}_I\nabla^{\parallel}_J\alpha}{\alpha}+\beta^+ \gamma_{IJ}\big)\lambda,
\end{array}\rd\quad \lb\begin{array}{l}
\tau^{11}=-\big( \gamma^{IJ}\frac{\nabla^{\parallel}_I\nabla^{\parallel}_J\alpha}{\alpha}+ \lc\Lambda\rc  \big)\lambda,\\
\tau^{1J}=\gamma^{IJ}\frac{\cp_{y^J}\alpha}{\alpha},\\
\tau^{IJ}=0.
\end{array}\rd
\end{equation}
where we have defined 
\begin{equation}
\beta^{\pm}:=\frac{\Lambda^{\pm}}{n-1} -\frac{(n-2)\vk}{2(r_0^{\pm})^2}.
\end{equation}
It is worth stressing that the constant curvature parameter $\vk$ does not appear explicitly in  \eqref{Ytauwarped}. It however appears implicitly in the metric $\gamma_{IJ}$ and in the corresponding covariant derivative $\nabla^{\parallel}$. 
In the next subsections we particularize further to (the maximally extended) Schwarzschild and Schwarzschild-de Sitter spacetimes.

\subsection{Schwarzschild spacetime}\label{secSchwarzschild}
If the metrics on both sides are Schwarzschild we have $\Lambda^{\pm}=0$ and $\vk=1$. By \eqref{matchcondsym}, the radii $r_0^{\pm}$ must coincide, so the Schwarzschild mass of both sides is necessarily the same. We write  $r_0$ instead of $r_0^{\pm}$ from now on. Thus,  \eqref{Ytauwarped} reduces to
\begin{equation}
\label{YSchwarzschild}
\hspace{-0.2cm}\lb\begin{array}{ll}
Y^-_{11}=0, &  Y^+_{11}=0,\\
Y^-_{1J}=0, & Y^+_{1J}=\frac{\cp_{y^J}\alpha}{\alpha},\\
Y^-_{IJ}=-\frac{(n-2)\lambda}{2r_0^2} \gamma_{IJ} , &  Y^+_{IJ}=\big(\frac{\nabla^{\parallel}_I\nabla^{\parallel}_J\alpha}{\alpha}-\frac{n-2}{2r_0^2} \gamma_{IJ}\big)\lambda,
\end{array}\rd\quad \lb\begin{array}{l}
\tau^{11}=-\gamma^{IJ}\frac{\nabla^{\parallel}_I\nabla^{\parallel}_J\alpha}{\alpha} \lambda,\\
\tau^{1J}=\gamma^{IJ}\frac{\cp_{y^J}\alpha}{\alpha},\\
\tau^{IJ}=0.
\end{array}\rd
\end{equation}

The tensor $\gamma$ is the round metric of radius $r_0$ so its Laplace-Beltrami operator is $r_0^{-2} \Delta_{\mathbb{S}^{n-1}}$ where $\Delta_{\mathbb{S}^{n-1}}$ is the Laplacian of the  unit $(n-1)$-sphere.  

For each natural number $l$ we let $\{ Y_{l,m} \}$, $m=0,...,N(n,l)-1$ be a collection of linearly independent solutions of 
\begin{equation}
\label{DeltaY=cY}
\Delta_{\mathbb{S}^{n-1}}Y_{l,m}=-l(l+n-2)Y_{l,m}
\end{equation}
which, as usual we call spherical harmonics. The number $N(n,l)$ is (see e.g. \cite{efthimiou2014spherical}) 
\begin{equation}
\lb\begin{array}{ll}
N(n,l)=1 & \text{if } l=0,\\
N(n,l)=\binom{l+n-2}{l}+\binom{l+n-3}{l-1} & \text{otherwise} .
\end{array}\rd
\end{equation}
Since $\{Y_{l,m}\}$ form a complete set of functions over $\mathbb{S}^{n-1}$, any (sufficiently regular) function $\alpha$ can be decomposed in this basis. Observe that $\alpha$ can be ensured to be positive by selecting the coefficient of $Y_{0,0}$ suitably positive and large. By expressing $\alpha$ as
\begin{equation}
\label{alphaSch}\alpha=\sum_{l=0}^{\infty}\sum_{m=0}^{N(n,l)-1}a_{l,m}Y_{l,m},\quad\text{where}\quad a_{l,m}\in\mathbb{R},
\end{equation}
the energy density of the shell is given by (cf. \eqref{YSchwarzschild})
\begin{align}
\label{energydensitySchwarzschild}\rho&=-\dfrac{\Delta_{\mathbb{S}^{n-1}}\alpha}{r_0^2\alpha}\lambda=\dfrac{\sum_{l=0}^{\infty}l(l+n-2)\sum_{m=0}^{N(n,l)-1}a_{l,m}Y_{l,m}}{r_0^2\sum_{l=0}^{\infty}\sum_{m=0}^{N(n,l)-1}a_{l,m}Y_{l,m}}\lambda.
\end{align}

The simplest case occurs when $\alpha$ is a positive constant. Then $[Y] =0$ and we have complete absence of shell. The step function $H = \alpha \lambda$ can be absorbed in the coordinates of the $\lp \mathcal{M}^{+}, g^{+}\rp$ side. This coordinate freedom is a consequence precisely  of the fact that  Schwarzschild admits a one-parameter isometry group leaving the Killing horizon, and its generators, invariant. This ensures that, in the absence of shell, we recover the global Schwarzschild spacetime, as we must.

We conclude with a simple but not trivial example in dimension four (i.e. $n=3$). Take
\begin{align}
\label{alphaSch}
\alpha(\theta)&= 3\sqrt{\pi}Y_{0,0}+\dfrac{2\sqrt{\pi}}{\sqrt{5}}Y_{2,0}(\theta)= 1+\dfrac{3}{2} \cos^2\theta,\qquad\text{where}\qquad\lb
\begin{array}{l}
Y_{0,0}=\frac{1}{2\sqrt{\pi}},\\ [\medskipamount]
Y_{2,0}=\frac{\sqrt{5}}{2\sqrt{\pi}}P_2(\cos\theta)
\end{array}\rd
\end{align}
and $P_l(x)$ denote Legendre polynomials of degree $l$. This yields energy density and energy fluxes
\begin{align}
\rho&=\dfrac{3\lp 3\cos^2\theta-1\rp}{ 1+\frac{3}{2} \cos^2\theta}\dfrac{\lambda}{r_0^2}\qquad j^{\theta}=-\dfrac{3\sin\theta\cos\theta}{r_0^2\lp 1+\frac{3}{2} \cos^2\theta\rp},\qquad j^{\phi}=0.
\end{align}
In Figure \ref{figSch} we plot the functions $\alpha(\theta)$, $j^{\theta}(\theta)$ and the energy density $\rho(\lambda,\theta)$ for $\lambda=1$, $\lambda=0$ and $\lambda=-1$ in units where $r_0=1$. As already discussed, the energy density changes sign at the bifurcation surface, despite the fact that the energy flux is constant along each null generator, including at the crossing of the bifurcation surface. The figure shows clearly how the energy flux component $j^{\theta}$ is positive (resp. negative) whenever the function $\alpha$ decreases (resp. increases), i.e. the energy flows towards those null generators with higher values of $\alpha$.

\begin{figure}[t]
\centering
\psfrag{a1}{$\alpha\lp \theta\rp$}
\psfrag{r1}{$\rho\lp -1,\theta\rp$}
\psfrag{j1}{$j^{\theta}\lp \theta\rp\qquad$}
\psfrag{a2}{$\alpha\lp \theta\rp$}
\psfrag{r2}{$\rho\lp 0,\theta\rp$}
\psfrag{j2}{$j^{\theta}\lp \theta\rp\qquad$}
\psfrag{a3}{$\alpha\lp \theta\rp$}
\psfrag{r3}{$\rho\lp 1,\theta\rp$}
\psfrag{j3}{$j^{\theta}\lp \theta\rp\qquad$}
\psfrag{s}{}
\psfrag{x}{}
\psfrag{C}{}
\psfrag{L}{$\theta$}
\includegraphics[width=15.5cm]{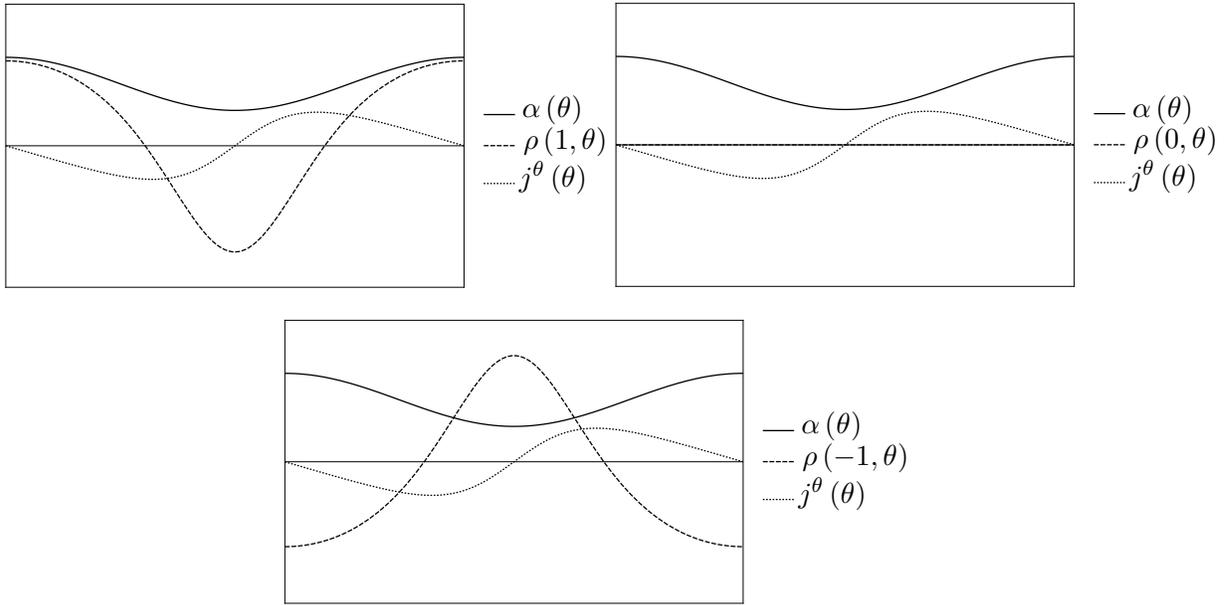}
\caption{For $r_0=1$ and $\alpha\lp\theta\rp$ given by \eqref{alphaSch}, the up-left, up-right and bottom plots show $\alpha(\theta)$, $j^{\theta}(\theta)$ and the energy density $\rho(\lambda,\theta)$ for $\lambda=1$, $\lambda=0$ and $\lambda=-1$ respectively.}
\label{figSch}
\end{figure}

\subsection{Schwarzschild-de Sitter spacetime}
Our final example is the matching of two Schwarzschild-de Sitter spacetimes. As before, $\varkappa=1$ and the horizon radii $r_0^{\pm}$ in both sides are forced to be the same so we can simply write $r_0$. A Schwarzschild-de Sitter spacetime of mass $m$ and cosmological constant $\Lambda>0$ may have several, one or none Killing horizons located at $r^{H}_{i}$ depending on the number of (positive) roots of the polynomial
\begin{equation}
0=P(r^{H}_{i})=(r^{H}_{i})^{n-2}-\dfrac{2\Lambda (r^{H}_{i})^{n}}{n(n-1)}-\dfrac{2m}{(n-1)(n-2)}.
\end{equation}
Since we do the matching on a preselected horizon we change the point of view, namely we take any positive value $r_0$ and use this expression to determine 
$m$ in terms of $r_0$ and $\Lambda$. The cosmological constants $\Lambda^{\pm}$ on both sides are allowed to be different. However, once they are selected, the corresponding masses $m^{\pm}$ must have jump 
\begin{equation}
\label{jumpmasses}[m] =  -\dfrac{(n-2)}{n}[\Lambda] r_0^{n}.
\end{equation}
Thus, a priori one can match across a horizon two Schwarzschild-de Sitter spacetimes with different masses and cosmological constants but \textit{only} if the parameters are related by \eqref{jumpmasses}.

The matter content of the shell is in this case given by \eqref{Ytauwarped}.  As in the previous section we may decompose the function $\alpha$ in terms of spherical harmonics. The corresponding expression for the energy density is now
\begin{align}
\label{energydensitySchwarzschild}\rho&=\lp\dfrac{\sum_{l=0}^{\infty}l(l+n-2)\sum_{m=0}^{N(n,l)-1}a_{l,m}Y_{l,m}}{r_0^2\sum_{l=0}^{\infty}\sum_{m=0}^{N(n,l)-1}a_{l,m}Y_{l,m}}-[\Lambda]\rp\lambda.
\end{align}
Let us conclude with some interesting observations. The first is that it is impossible to construct a (non-trivial) shell with vanishing energy density. This is because in such case it must hold
\begin{equation}
\label{rhozero0}\Delta_{\mathbb{S}^{n-1}}\alpha=-r_0^{2}\lc\Lambda\rc \alpha.
\end{equation}
and all solutions of these equation must necessarily have zeroes, which is not allowed for the matching function $\alpha$.

An interesting example is when the shell is composed on null dust, i.e. with identically zero energy-flux. By \eqref{Ytauwarped}, this requires $\alpha$ to be a (positive) constant and then the energy density of the null dust is
\begin{equation}
\rho=-[\Lambda]\lambda.
\end{equation}
The behaviour of this null dust is striking. Assume $[\Lambda] <0 $ for definiteness. Then the energy density is everywhere positive in the past of the bifurcation surface (i.e. for $\lambda<0$) so the system starts being perfectly reasonable from a physical point of view. The shell then evolves on its own in a manner dictated by the field equations and ends up in a state in which the energy density is everywhere negative. This negative energy density cannot be considered as unphysical, since it has evolved from a physically reasonable initial state, the system itself is physically reasonable (a collection of incoherent massless particles) and its evolution is dictated by the gravitational shell equations that follow from the Einstein field equations. This is a rather surprising behaviour.

Furthermore, this behaviour is not exclusive of null dust. In fact, this occurs for more general functions $\alpha$. It suffices to select $\alpha$ to be an everywhere positive function on $\mathbb{S}^{n-1}$ and, for $\lambda<0$, the energy density $\rho$ will always be positive provided that the jump $[\Lambda]$ is suitably positive and large. Then, we have a shell  with initial positive energy density and non-zero energy flux which unavoidably evolves into a state of negative energy density with no change along the evolution of the energy flux, which by \eqref{Ytauwarped} is independent of $\lambda$.

\appendix

\section{Pullback to $\Sigma$ of covariant derivatives on $\Omega^{\pm}$}\label{apppullbacks}
In this Appendix we establish a relationship, needed in Section \ref{secRestrictionY}, between  the pullback of covariant derivatives along the sections of $\Omega^{\pm}$ and covariant derivatives along the sections of $\Sigma$. We do this under the assumption that $\bs{\chi}^{k^{\pm}}_{\pm}=0$ which is the case of interest in this paper. The result is as follows.
\begin{lem}\label{lempullback}In the setup of Section \ref{secPreliminaries} (in particular for covariant tensors $\bs{T}^{\pm}$ on $\Omega^{\pm}$ satisfying $\bs{T}^{\pm} (..., k^{\pm}, ...)=0$), if the second fundamental form $\bs{\chi}_+^{k^+}$ vanishes everywhere on $\Omega^+$, the derivative operators $\nabla^{\parallel}$ and $\accentset{\circ}{\nabla}^{\pm}$ satisfy the identities
\begin{align}
\label{nablacircvsnablaparallel1}\nabla^{\parallel}_{I}\Tm_{A_1...A_p}&=\accentset{\circ}{\nabla}_{v^-_I}^-T^-_{A_1...A_p},\\
\label{nablacircvsnablaparallel2}\nabla^{\parallel}_{I}\Tp_{A_1...A_p}&=b_{A_1}^{B_1}...b_{A_p}^{B_p}\lp (\nabla^{\parallel}_IH)k^+( T^+_{B_1...B_p})+b_I^J\accentset{\circ}{\nabla}_{v^+_{J}}^+T^+_{B_1...B_p}\rp.
\end{align}
In particular, when $\bs{T}^{+}$ is such that $k^+(T^{+}_{A_1...A_p})=0$, one obtains
\begin{align}
\label{nablaswithk()zero}
\nabla^{\parallel}_{I}\Tp_{A_1...A_p}=b_I^Jb_{A_1}^{B_1}...b_{A_p}^{B_p}\accentset{\circ}{\nabla}_{v^+_{J}}^+T^+_{B_1...B_p}.
\end{align}
\end{lem} 
\begin{proof}
Throughout this proof, we view $b_I^J$ and $h^{\pm}_{AB}:=g^{\pm}(v^{\pm}_A,v^{\pm}_B)$ as scalars defined on the boundaries $\Omega^{\pm}$. Their pullbacks to $\Sigma$ are denoted with the same symbol and we let the context determine the precise meaning. The Levi-Civita covariant derivative $\nabla^{\parallel}$ has Christoffel symbols 
\begin{align}
\label{christparallel} \Gamma^{\parallel}{}^{B}_{JI}&=\frac{1}{2}\gamma^{BA}\lp\cp_{y^I}\gamma_{AJ}+\cp_{y^J}\gamma_{AI}-\cp_{y^A}\gamma_{IJ}\rp,
\end{align}
where $\gamma_{IJ}$ denotes the components of the induced metric on the sections $\{\lambda=\text{const.}\}\subset\Sigma$ in the coordinates $\{y^A\}$.

From \eqref{emhd}, we know that $\gamma_{IJ}:=g^{\pm}(e^{\pm}_I,e^{\pm}_J)$ which, together with \eqref{evectors1}-\eqref{evectors2}, yields 
\begin{equation}
\label{gammaandhIJ}\gamma_{IJ}=h_{IJ}^-\qquad\text{and}\qquad\gamma_{IJ}=b_I^Ab_J^Bh^+_{AB}\quad\text{(hence }\gamma^{BA}=(b^{-1})^B_L(b^{-1})^A_Kh^{LK}_+\text{)} 
\end{equation}
respectively. The first, together with \eqref{evectors1}, immediately gives \eqref{nablacircvsnablaparallel1}.

To prove \eqref{nablacircvsnablaparallel2}, we particularize \eqref{k(hIJ)} and \eqref{christoffels,s=const} for $\bs{\chi}^{k^+}_+=0$ and use that $\cp_{y^I}b_A^C=\cp_{y^A}b_I^C$ (cf. \eqref{b_I^J}) together with 
\begin{align}
\cp_{y^A}h^+_{CD}&=e^+_A(h^+_{CD})=a_Ak^+(h^+_{CD})+b_A^Lv^+_L(h^+_{CD})=b_A^Lv^+_L(h^+_{CD}).
\end{align}
Substituting $\gamma_{IJ}$ from \eqref{gammaandhIJ} into \eqref{christparallel}, one obtains
\begin{align}
\nonumber \Gamma^{\parallel}{}^{B}_{JI}=&\spc\frac{1}{2}(b^{-1})^B_L(b^{-1})^A_Kh^{LK}_+\lp\cp_{y^I}(b_A^Cb_J^Dh^+_{CD})+\cp_{y^J}(b_A^Cb_I^Dh^+_{CD})-\cp_{y^A}(b_I^Cb_J^Dh^+_{CD})\rp\\
\nonumber =&\spc\frac{1}{2}(b^{-1})^B_Lh^{LK}_+\Big(
2(\cp_{y^I}b_J^D)h^+_{KD}+b_J^D(\cp_{y^I}h^+_{KD})+b_I^D(\cp_{y^J}h^+_{KD})-(b^{-1})^A_Kb_I^Cb_J^D(\cp_{y^A}h^+_{CD})\Big)\\
\nonumber =&\spc(b^{-1})^B_L
(\cp_{y^I}b_J^L)+\frac{1}{2}(b^{-1})^B_Lb_I^Cb_J^Dh^{LK}_+\Big(
v^+_C(h^+_{KD})+v^+_D(h^+_{KC})-v^+_K(h^+_{CD})\Big)\\
\label{relchristoffel} =&\spc(b^{-1})^B_L\lp
(\cp_{y^I}b_J^L)+b_I^Cb_J^D\chr^+{}_{CD}^L\rp,
\end{align} 
where in the last equality we inserted \eqref{christoffels,s=const}. Expanding the derivative $\cp_{y^I}\Tp_{A_1...A_p}$ as (cf. \eqref{tensorremark2})
\begin{equation}
\cp_{y^I}\Tp_{A_1...A_p}= b_{A_1}^{B_1}...b_{A_p}^{B_p}\cp_{y^I}(T^+_{B_1...B_p})+\sum_{i=1}^pb_{A_1}^{B_1}...b^{B_{i-1}}_{A_{i-1}}b^{B_{i+1}}_{A_{i+1}}...b_{A_p}^{B_p}(\cp_{y^I}b^{L}_{A_i})T^+_{B_1...B_{i-1}LB_{i+1}...B_p}
\end{equation}
and inserting it into $\nabla^{\parallel}_{I}\Tp_{A_1...A_p}:=\cp_{y^I}\Tp_{A_1...A_p}-\sum_{i=1}^p\Gamma^{\parallel}{}^{B}_{A_iI}\Tp_{A_1...A_{i-1}BA_{i+1}...A_p}$ gives
\begin{align}
\label{bydib} \nabla^{\parallel}_{I}\Tp_{A_1...A_p}= b_{A_1}^{B_1}...b_{A_p}^{B_p}\lp\cp_{y^I}(T^+_{B_1...B_p})-\sum_{i=1}^p b_I^C\chr^+{}_{CB_{i}}^KT^+_{B_1...B_{i-1}KB_{i+1}...B_p}\rp,
\end{align}
after taking \eqref{relchristoffel} into account. By \eqref{evectors2}-\eqref{b_I^J}, we can rewrite $\cp_{y^I}T^+_{B_1...B_p}$ as 
\begin{equation}
\nonumber \cp_{y^I}T^+_{B_1...B_p}=e_I^+(T^+_{B_1...B_p})=(\nabla^{\parallel}_IH)k^+(T^+_{B_1...B_p})+b_I^Cv_C^+(T^+_{B_1...B_p}),
\end{equation}
so \eqref{bydib} becomes
\begin{align}
\nonumber \nabla^{\parallel}_{I}\Tp_{A_1...A_p}:=&\spc b_{A_1}^{B_1}...b_{A_p}^{B_p}\bigg((\nabla^{\parallel}_IH)k^+(T^+_{B_1...B_p})\\
\label{resnablapullback} &\spc +b_I^C\Big( v_C^+(T^+_{B_1...B_p})-\sum_{i=1}^p \chr^+{}_{CB_{i}}^KT^+_{B_1...B_{i-1}KB_{i+1}...B_p}\Big)\bigg).
\end{align}
The coefficients $\chr^+{}^{K}_{AB}$ are obviously symmetric in the indices $A,B$ (see \eqref{christoffels,s=const}) and it holds 
\begin{align}
\nonumber v^+_{I}\lp h^+_{KJ}\rp-\chr^+{}^{L}_{KI}h^+_{LJ}-\chr^+{}^{L}_{JI}h^+_{KL}&=\nabla_{v^+_{I}}\lp \la v^+_K,v^+_J\rag\rp-\chr^+{}^{L}_{KI}h^+_{LJ}-\chr^+{}^{L}_{JI}h^+_{KL}\\
\nonumber &= \la \nabla_{v^+_{I}}v^+_K,v^+_J\rag+\la \nabla_{v^+_{I}}v^+_J,v^+_K\rag-\chr^+{}^{L}_{KI}h^+_{LJ}-\chr^+{}^{L}_{JI}h^+_{KL}\\
\label{derh}&{=} -\frac{1}{\nfi^+}\lp \bs{\chi}_+^{k^+}\lp v^+_I,v^+_K\rp\psi^+_J+\bs{\chi}_+^{k^+}\lp v^+_I,v^+_J\rp\psi^+_K\rp=0,
\end{align}
where in the last line we used \eqref{covderinfo1} and that $\bs{\chi}^{k^{\pm}}_{\pm}=0$. This means that, for any covariant tensor field on $\Omega^{+}$
of the form $\bs{Q}^+=Q^+{}_{I_1...I_p}\bs{\omega}_+^{I_1}\otimes...\otimes\bs{\omega}_+^{I_p}$ (i.e. satisfying that $\bs{Q}^+(...,k^+,...)=0$) the operator $D^+$ defined by 
\begin{align}
\nonumber D^+_{v^+_K}Q^+{}_{I_1...I_p}^{J_1...J_q}&:=v^+_K\lp Q^+{}_{I_1...I_p}^{J_1...J_q}\rp-\sum_{i=1}^{p}\chr^+{}_{I_iK}^{L}Q^+{}_{I_1...I_{i-1}LI_{i+1}...I_p}^{J_1...J_q}
\end{align}
is the Levi-Civita covariant derivative on the sections $\{s^+=\text{const.}\}$ (i.e. $D^+\equiv\accentset{\circ}{\nabla}^+$). This together with \eqref{resnablapullback} proves \eqref{nablacircvsnablaparallel2} as well as \eqref{nablaswithk()zero}.
\end{proof}

\section{Ricci tensors on non-degenerate Killing horizons}\label{appedixA}
In this appendix, we compute explicit expressions for the Ricci tensors of the ambient spacetimes when the boundaries are non-degenerate Killing horizons. The results are not new and are included for the sake of completeness. We do the computation in R\'acz-Wald coordinates which yields the results quite directly. For a  different, more geometric, approach we refer to \cite{gourgoulhon20063+}.

We work in  R\'acz-Wald coordinates $\{u,v,x^A\}$ in a neighbourhood of a non-degenerate boundary $\chor$ with a bifurcation surface $\mathcal{S}$ as described in Section \ref{secRestrictionY} (in particular, without loss of generality the function in the metric \eqref{metricRW} has been chosen to be constant on the horizon $\chor$). The identity $\nabla_{\alpha}\nabla_{\beta}\xi^{\mu}+{R^{\mu}}_{\beta\nu\alpha}\xi^{\nu}=0$, which can be explicitly written as 
\begin{equation}
\cp_{\alpha}\cp_{\beta}\xi^{\mu}+\xi\big( \Gamma^{\mu}_{\alpha\beta}\big) +\Gamma^{\mu}_{\rho\beta} \cp_{\alpha}\xi^{\rho}+\Gamma^{\mu}_{\rho\alpha} \cp_{\beta}\xi^{\rho}-\Gamma^{\rho}_{\beta\alpha} \cp_{\rho}\xi^{\mu}=0,
\end{equation}
simplifies to 
\begin{equation}
\label{eqsappGammas}
\xi\big( \Gamma^{\mu}_{\alpha\beta}\big)+\ke\big(\Gamma^{\mu}_{v\beta}\delta_{\alpha}^v-\Gamma^{\mu}_{u\beta}\delta_{\alpha}^u+\Gamma^{\mu}_{v\alpha}\delta_{\beta}^v-\Gamma^{\mu}_{u\alpha}\delta_{\beta}^u-\Gamma^{v}_{\beta\alpha}\delta^{\mu}_v+\Gamma^{u}_{\beta\alpha}\delta^{\mu}_u\big)=0
\end{equation}
when using \eqref{killingRW}, which yields $\cp_{\alpha}\xi^{\beta}=\ke\big(\delta_{\alpha}^v\delta_{v}^{\beta}-\delta_{\alpha}^u\delta_{u}^{\beta}\big)$, $\cp_{\alpha}\cp_{\beta}\xi^{\mu}=0$. The set of equations \eqref{eqsappGammas} constitutes a hierarchical system of ODE which one easily solves as 
\begin{align}
\label{appeq1} \hspace{-0.3cm}\Gamma^{u}_{uu}&=v \mrg^{u}_{uu},  & \Gamma^{u}_{uv}&=u \mrg^{u}_{uv}, & \Gamma^{u}_{uA}&=\mrg^{u}_{uA}, & \Gamma^{u}_{vA}&=u^2\mrg^{u}_{vA}, & \Gamma^{u}_{AB}&=u \mrg^{u}_{AB}, \\
\label{appeq2} \hspace{-0.3cm}\Gamma^{v}_{vv}&=u \mrg^{v}_{vv}, & \Gamma^{v}_{vu}&=v \mrg^{v}_{vu}, & \Gamma^{v}_{vA}&=\mrg^{v}_{vA}, & \Gamma^{v}_{uA}&=v^2\mrg^{v}_{uA}, & \Gamma^{v}_{AB}&=v \mrg^{v}_{AB},   \\
\label{appeq3}\hspace{-0.3cm} \Gamma^{A}_{Bu}&=v \mrg^{A}_{Bu}, & \Gamma^{A}_{Bv}&=u \mrg^{A}_{Bv}, & \Gamma^{A}_{BC}&=\mrg^{A}_{BC}, &   &   
\end{align}
where $\mrg^{\mu}_{\alpha\beta}$ are functions depending on $\{ \omega:=uv,x^A\}$. Besides, one finds
\begin{align}
\hspace{-0.5cm}\lb\begin{array}{l}
\Gamma^u_{uA}\vert_{\chor}=-\frac{1}{2G}\lp-\cp_{A}G+\cp_{u}g_{v A}\rp\vert_{\chor}=-\frac{1}{2G}\cp_{u}g_{v A}\vert_{\chor}\\[\medskipamount]
\Gamma^v_{vA}\vert_{\chor}=-\frac{1}{2G}\lp-\cp_{A}G-\cp_{u}g_{vA}\rp\vert_{\chor}=\frac{1}{2G}\cp_{u}g_{v A}\vert_{\chor}
\end{array}\rd\spc\Longrightarrow\spc \Gamma^u_{uA}\vert_{\chor}=-\Gamma^v_{vA}\vert_{\chor}=\frac{1}{2}w_A\vert_{\chor},
\end{align}
which, together with $\cp_{v}\mrg^{\alpha}_{\mu\nu}\vert_{\chor}=u\cp_{\omega}\mrg^{\alpha}_{\mu\nu}\vert_{\chor}=0$ and \eqref{appeq1}-\eqref{appeq3} yields 
\begin{align}
\label{eqapp4}\hspace{-0.2cm}\cp_{\mu}\Gamma^{\mu}_{ B A}\spc\stackbin{\chor}{=}&\spc \cp_{u}\Gamma^{u}_{ B A}+\cp_{v}\Gamma^{v}_{ B A}+\cp_{C}\Gamma^{C}_{ B A}\stackbin{\chor}{=}\mrg^{u}_{BA}+\mrg^{v}_{BA}+\cp_{D}\mrg^{D}_{ B A},\\
\label{eqapp5}\hspace{-0.2cm}\cp_{ B}\Gamma^{\mu}_{\mu A}\spc\stackbin{\chor}{=}&\spc \cp_{ B}\mrg^{u}_{u A}+\cp_{ B}\mrg^{v}_{v A}+\cp_{ B}\mrg^{D}_{D A}\stackbin{\chor}{=}\cp_{ B}\mrg^{D}_{D A}\\
\label{eqapp6}\hspace{-0.2cm}\Gamma^{\mu}_{\mu \nu}\Gamma^{\nu}_{ B A}\spc\stackbin{\chor}{=}&\spc \lp\Gamma^{u}_{uD}+\Gamma^{v}_{vD}+\Gamma^{C}_{CD}\rp\Gamma^{D}_{ B A}\stackbin{\chor}{=}\mrg^{C}_{CD}\mrg^{D}_{ B A}\\
\label{eqapp7}\hspace{-0.2cm}\Gamma^{\mu}_{ B\nu}\Gamma^{\nu}_{\mu A}\spc\stackbin{\chor}{=}&\spc \Gamma^{u}_{ Bu}\Gamma^{u}_{u A} +\Gamma^{v}_{ Bv}\Gamma^{v}_{v A}+\Gamma^{C}_{ BD}\Gamma^{D}_{CA}\stackbin{\chor}{=}2\Gamma^{u}_{ Bu}\Gamma^{u}_{u A}+\Gamma^{C}_{ BD}\Gamma^{D}_{CA}\stackbin{\chor}{=}\frac{1}{2}w_Aw_B+\mrg^{C}_{ BD}\mrg^{D}_{CA}.
\end{align}
Let $\accentset{\circ}{\nabla}$ be the Levi-Civita covariant derivative on the sections $\{ v=\text{const.}\}$, i.e. defined with respect to $h_{AB}:=\ovl{\gamma}_{AB}\vert_{\chor}\equiv\ovl{\gamma}_{AB}\lp 0,x^J\rp$ and $\accentset{\circ}{R}_{AB}$ its Ricci tensor. Obviously these objects only depend on the coordinates $\{x^A\}$, i.e. are independent of the section $\{ v = \text{const.}\}$. From \eqref{appeq1}-\eqref{appeq2} it follows
\begin{align}
\nonumber \Gamma^{v}_{BA}=&\spc\dfrac{1}{2G}\cp_{u}\ovl{\gamma}_{AB}=\dfrac{v}{2G}\cp_{\omega}\ovl{\gamma}_{AB}\qquad\Longrightarrow\qquad \mrg^{v}_{BA}\stackbin{\chor}{=}\frac{1}{2G}\cp_{\omega}\ovl{\gamma}_{AB},\\
\nonumber \Gamma^{u}_{BA}=&\spc-\dfrac{1}{2G}\lp \cp_{A}g_{vB}+\cp_{B}g_{vA}-\cp_{v}\ovl{\gamma}_{AB}\rp+\dfrac{1}{2}g^{uD}\lp \cp_{A}\ovl{\gamma}_{DB}+\cp_{B}\ovl{\gamma}_{DA}-\cp_{D}\ovl{\gamma}_{AB}\rp-\dfrac{1}{2}g^{uu}\cp_{u}\ovl{\gamma}_{AB}\\
\nonumber \Longrightarrow\quad\mrg^{u}_{BA}\stackbin{\chor}{=}&\spc\frac{1}{2}\big(\accentset{\circ}{\nabla}_{A}w_B+\accentset{\circ}{\nabla}_{B}w_A\big) +\frac{1}{2G}\cp_{\omega}\ovl{\gamma}_{AB},
\end{align}
from where one concludes
\begin{align}
\nonumber R_{AB}\spc\stackbin{\chor}{=}&\spc \cp_{\mu}\Gamma^{\mu}_{ B A}-\cp_{ B}\Gamma^{\mu}_{\mu A}+\Gamma^{\mu}_{\mu \nu}\Gamma^{\nu}_{ B A}-\Gamma^{\mu}_{ B\nu}\Gamma^{\nu}_{\mu A}\\
\nonumber \stackbin{\chor}{=}&\spc\cp_C\mrg^{C}_{BA}-\cp_B\mrg^{D}_{DA}+\mrg^{D}_{DC}\mrg^{C}_{BA}-\mrg^{C}_{BD}\mrg^{D}_{CA}+\mrg^{u}_{BA}+\mrg^{v}_{BA}-\frac{1}{2}w_Aw_B\\
\label{Riccinondeg2}\spc\stackbin{\chor}{=}&\spc \accentset{\circ}{R}_{AB}+\dfrac{1}{2}\lp \accentset{\circ}{\nabla}_{A}w_B+\accentset{\circ}{\nabla}_{B}w_A \rp+\dfrac{1}{G}\cp_{\omega}\ovl{\gamma}_{AB}- \dfrac{1}{2}w_Aw_B.
\end{align}
Observe that the right-hand side of \eqref{Riccinondeg2} is independent of $v$, so the same is true for the components $R_{AB}$ of the ambient Ricci tensor at the Killing horizon. 

\section*{Acknowledgements}
The authors acknowledge financial support under the projects PGC2018-096038-B-I00 (Spanish Ministerio de Ciencia e Innovaci{\'o}n and FEDER ``A way of making Europe") and  SA096P20 (JCyL). M. Manzano also acknowledges the Ph.D. grant FPU17/03791 (Spanish Ministerio de Ciencia, Innovaci{\'o}n y Universidades).

\begingroup
\let\itshape\upshape

\bibliographystyle{acm}

\bibliography{ref}

\begin{thebibliography}{10}

\bibitem{ashtekar2000generic}
{\sc Ashtekar, A., Beetle, C., Dreyer, O., Fairhurst, S., Krishnan, B.,
  Lewandowski, J., and Wi{\'s}niewski, J.}
\newblock ``{G}eneric isolated horizons and their applications".
\newblock {\em Physical Review Letters \textbf{85}}, 17 (2000), 3564.

\bibitem{ashtekar2002geometry}
{\sc Ashtekar, A., Beetle, C., and Lewandowski, J.}
\newblock ``{G}eometry of generic isolated horizons".
\newblock {\em Classical and Quantum Gravity \textbf{19}}, 6 (2002), 1195.

\bibitem{ashtekar2000isolated}
{\sc Ashtekar, A., Fairhurst, S., and Krishnan, B.}
\newblock ``{I}solated horizons: {H}amiltonian evolution and the first law".
\newblock {\em Physical Review D \textbf{62}}, 10 (2000), 104025.

\bibitem{bardeen1973four}
{\sc Bardeen, J.~M., Carter, B., and Hawking, S.~W.}
\newblock ``{T}he four laws of black hole mechanics".
\newblock {\em Communications in mathematical physics \textbf{31}}, 2 (1973),
  161--170.

\bibitem{barrabes2003singular}
{\sc Barrab{\'e}s, C., and Hogan, P.~A.}
\newblock {\em ``{S}ingular null hypersurfaces in general relativity:
  light-like signals from violent astrophysical events"}.
\newblock World Scientific, (2003).

\bibitem{barrabes1991thin}
{\sc Barrab{\'e}s, C., and Israel, W.}
\newblock ``{T}hin shells in general relativity and cosmology: the lightlike
  limit".
\newblock {\em Physical Review D \textbf{43}\/} (1991), 1129--1142.

\bibitem{bhattacharjee2020memory}
{\sc Bhattacharjee, S., and Kumar, S.}
\newblock ``{M}emory effect and {B}{M}{S} symmetries for extreme black holes".
\newblock {\em Physical Review D \textbf{102}\/} (2020), 044041.

\bibitem{binetruy2018closed}
{\sc Bin{\'e}truy, P., Helou, A., and Lamy, F.}
\newblock ``{C}losed trapping horizons without singularity".
\newblock {\em Physical Review D \textbf{98}\/} (2018), 064058.

\bibitem{bonnor1981junction}
{\sc Bonnor, W.~B., and Vickers, P.~A.}
\newblock ``{J}unction conditions in general relativity".
\newblock {\em General Relativity and Gravitation \textbf{13}}, 1 (1981),
  29--36.

\bibitem{carballo2021inner}
{\sc Carballo~Rubio, R., Di~Filippo, F., Liberati, S., Pacilio, C., and Visser,
  M.}
\newblock ``{I}nner horizon instability and the unstable cores of regular black
  holes".
\newblock {\em arXiv:2101.05006\/} (2021).

\bibitem{chapman2018holographic}
{\sc Chapman, S., Marrochio, H., and Myers, R.~C.}
\newblock ``{H}olographic complexity in {V}aidya spacetimes. {P}art {I}{I}".
\newblock {\em Journal of High Energy Physics \textbf{2018}\/} (2018), 1--103.

\bibitem{chen2011pseudo}
{\sc Chen, B.}
\newblock {\em ``{P}seudo-{R}iemannian geometry, $\delta$-invariants and
  applications"}.
\newblock World Scientific, (2011).

\bibitem{clarke1987junction}
{\sc Clarke, C. J.~S., and Dray, T.}
\newblock ``{J}unction conditions for null hypersurfaces".
\newblock {\em Classical and Quantum Gravity \textbf{4}\/} (1987), 265.

\bibitem{darmois1927memorial}
{\sc Darmois, G.}
\newblock ``{L}es {\'e}quations de la gravitation einsteinienne".
\newblock {\em M{\'e}morial des Sciences Math{\'e}matiques, Fascicule XXV
  (Paris: Gauthier-Villars) \textbf{44}\/} (1927).

\bibitem{efthimiou2014spherical}
{\sc Efthimiou, C., and Frye, C.}
\newblock {\em ``{S}pherical harmonics in $p$ dimensions"}.
\newblock World Scientific, (2014).

\bibitem{fairoos2017massless}
{\sc Fairoos, C., Ghosh, A., and Sarkar, S.}
\newblock ``{M}assless charged particles: {C}osmic censorship, and the third
  law of black hole mechanics".
\newblock {\em Physical Review D \textbf{96}\/} (2017), 084013.

\bibitem{frolov2012black}
{\sc Frolov, V., and Novikov, I.}
\newblock {\em ``{B}lack hole physics: basic concepts and new developments"},
  vol.~\textbf{96}.
\newblock Springer Science \& Business Media, (2012).

\bibitem{galloway1999maximum}
{\sc Galloway, G.~J.}
\newblock ``{M}aximum principles for null hypersurfaces and null splitting
  theorems".
\newblock {\em Ann. Henri Poincaré 1\/} (2000).

\bibitem{galloway2004null}
{\sc Galloway, G.~J.}
\newblock ``{N}ull geometry and the {E}instein equations".
\newblock In {\em ``The Einstein equations and the large scale behavior of
  gravitational fields}. Springer, (2004), pp.~379--400.

\bibitem{gourgoulhon20063+}
{\sc Gourgoulhon, E., and Jaramillo, J.~L.}
\newblock ``{A} $3+1$ perspective on null hypersurfaces and isolated horizons".
\newblock {\em Physics Reports \textbf{423}\/} (2006), 159--294.

\bibitem{hajivcek1973exact}
{\sc H{\'a}ji{\v{c}}ek, P.}
\newblock ``{E}xact models of charged black holes".
\newblock {\em Communications in Mathematical Physics \textbf{34}}, 1 (1973),
  53--76.

\bibitem{israel1966singular}
{\sc Israel, W.}
\newblock ``{S}ingular hypersurfaces and thin shells in general relativity".
\newblock {\em Il Nuovo Cimento B \textbf{44}\/} (1966), 1--14.

\bibitem{jaramillo2009isolated}
{\sc Jaramillo, J.~L.}
\newblock ``{I}solated horizon structures in quasiequilibrium black hole
  initial data".
\newblock {\em Physical Review D \textbf{79}}, 8 (2009), 087506.

\bibitem{kay1991theorems}
{\sc Kay, B.~S., and Wald, R.~M.}
\newblock ``{T}heorems on the uniqueness and thermal properties of stationary,
  nonsingular, quasifree states on spacetimes with a bifurcate {K}illing
  horizon".
\newblock {\em Physics Reports \textbf{207}}, 2 (1991), 49--136.

\bibitem{kobayashi1958fixed}
{\sc Kobayashi, S.}
\newblock ``{F}ixed points of isometries".
\newblock {\em Nagoya Mathematical Journal \textbf{13}\/} (1958), 63--68.

\bibitem{kokubu2018energy}
{\sc Kokubu, T., Jhingan, S., and Harada, T.}
\newblock ``{E}nergy emission from a high curvature region and its
  backreaction".
\newblock {\em Physical Review D \textbf{97}\/} (2018), 104014.

\bibitem{le1955theories}
{\sc Lichnerowicz, A.}
\newblock {\em ``{T}h{\'e}ories relativistes de la gravitation et de
  l'{\'e}lectromagn{\'e}tisme"}.
\newblock Masson {\&} Cie, (1955).

\bibitem{manzano2021null}
{\sc Manzano, M., and Mars, M.}
\newblock ``{N}ull shells: general matching across null boundaries and
  connection with cut-and-paste formalism".
\newblock {\em Classical and Quantum Gravity\/} (2021).

\bibitem{mars2012stability}
{\sc Mars, M.}
\newblock ``{S}tability of {M}{O}{T}{S} in totally geodesic null horizons".
\newblock {\em Classical and Quantum Gravity \textbf{29}}, 14 (2012), 145019.

\bibitem{mars2013constraint}
{\sc Mars, M.}
\newblock ``{C}onstraint equations for general hypersurfaces and applications
  to shells".
\newblock {\em General Relativity and Gravitation \textbf{45}\/} (2013),
  2175--2221.

\bibitem{mars2020hypersurface}
{\sc Mars, M.}
\newblock ``{H}ypersurface data: general properties and {B}irkhoff theorem in
  spherical symmetry".
\newblock {\em Mediterranean Journal of Mathematics \textbf{17}\/} (2020),
  1--45.

\bibitem{mars2018multiple}
{\sc Mars, M., Paetz, T., and Senovilla, J. M.~M.}
\newblock ``{M}ultiple {K}illing horizons".
\newblock {\em Classical and Quantum Gravity \textbf{35}}, 15 (2018), 155015.

\bibitem{mars2018nearhorizon}
{\sc Mars, M., Paetz, T., and Senovilla, J. M.~M.}
\newblock ``{M}ultiple {K}illing horizons and near horizon geometries".
\newblock {\em Classical and Quantum Gravity \textbf{35}}, 24 (2018), 245007.

\bibitem{mars2019multiple}
{\sc Mars, M., Paetz, T., and Senovilla, J. M.~M.}
\newblock ``{M}ultiple {K}illing horizons: the initial value formulation for
  ${\Lambda}$-vacuum".
\newblock {\em Classical and Quantum Gravity \textbf{37}}, 2 (2019), 025010.

\bibitem{mars2007lorentzian}
{\sc Mars, M., Senovilla, J., and Vera, R.}
\newblock ``{L}orentzian and signature changing branes".
\newblock {\em Physical Review D \textbf{76}\/} (2007), 044029.

\bibitem{mars1993geometry}
{\sc Mars, M., and Senovilla, J. M.~M.}
\newblock ``{G}eometry of general hypersurfaces in spacetime: junction
  conditions".
\newblock {\em Classical and Quantum Gravity \textbf{10}\/} (1993), 1865.

\bibitem{nikitin2018stability}
{\sc Nikitin, I.}
\newblock ``{S}tability of white holes revisited".
\newblock {\em arXiv:1811.03368\/} (2018).

\bibitem{kobayashi1969vol2}
{\sc Nomizu, K., and Kobayashi, S.}
\newblock {\em ``{F}oundations of differential geometry, {V}ol. II"}.
\newblock New York: Interscience Publisher, (1969).

\bibitem{obrien1952jump}
{\sc O'Brien, S.}
\newblock ``{J}ump conditions at discontinuities in general relativity".
\newblock {\em Communications of the Dublin Institute for Advanced Studies
  \textbf{9}\/} (1952).

\bibitem{neil1983semi}
{\sc O'{N}eill, B.}
\newblock {\em ``{S}emi-{R}iemannian geometry with applications to
  relativity"}.
\newblock Academic press, (1983).

\bibitem{Penrose:1972xrn}
{\sc Penrose, R.}
\newblock ``{T}he geometry of impulsive gravitational waves".
\newblock In {\em {``{G}eneral {R}elativity: {P}apers in honour of {J}.{L}.
  {S}ynge"}}, L.~O'Raifeartaigh, Ed. (1972), pp.~101--115.

\bibitem{podolsky1999nonexpanding}
{\sc Podolsk{\`y}, J., and Griffiths, J.~B.}
\newblock ``{N}onexpanding impulsive gravitational waves with an arbitrary
  cosmological constant".
\newblock {\em Physics Letters A \textbf{261}\/} (1999), 1--4.

\bibitem{podolsky2019cut}
{\sc Podolsk{\`y}, J., S{\"a}mann, C., Steinbauer, R., and {\v{S}}varc, R.}
\newblock ``{C}ut-and-paste for impulsive gravitational waves with
  {${\Lambda}$}: the geometric picture".
\newblock {\em Physical Review D \textbf{100}\/} (2019), 024040.

\bibitem{podolsky2017penrose}
{\sc Podolsk{\`y}, J., {\v{S}}varc, R., Steinbauer, R., and S{\"a}mann, C.}
\newblock ``{P}enrose junction conditions extended: impulsive waves with
  gyratons".
\newblock {\em Physical Review D \textbf{96}\/} (2017), 064043.

\bibitem{poisson2002reformulation}
{\sc Poisson, E.}
\newblock ``{A} reformulation of the {B}arrabes-{I}srael null-shell formalism".
\newblock {\em arXiv:gr-qc/0207101\/} (2002).

\bibitem{poisson2004relativist}
{\sc Poisson, E.}
\newblock {\em ``{A} relativist's toolkit: the mathematics of black-hole
  mechanics"}.
\newblock Cambridge university press, (2004).

\bibitem{racz1992extensions}
{\sc R{\'a}cz, I., and Wald, R.~M.}
\newblock ``{E}xtensions of spacetimes with {K}illing horizons".
\newblock {\em Classical and Quantum Gravity \textbf{9}}, 12 (1992), 2643.

\bibitem{senovilla2018equations}
{\sc Senovilla, J. M.~M.}
\newblock ``{E}quations for general shells".
\newblock {\em Journal of High Energy Physics \textbf{2018}\/} (2018), 134.

\bibitem{wald1984general}
{\sc Wald, R.~M.}
\newblock ``{G}eneral relativity".
\newblock {\em Chicago, University of Chicago Press\/} (1984).

\end{thebibliography}

\end{document}